\documentclass[10pt, journal]{IEEEtran}
\usepackage[latin1]{inputenc}
\usepackage[T1]{fontenc}
\usepackage[english]{babel}

	\usepackage{amsmath}
	\usepackage{amsthm}
	\usepackage{mathtools}
	\usepackage{amssymb}
	\usepackage{dsfont}
	\usepackage{stmaryrd,color}
	\allowdisplaybreaks[4]
    \usepackage{subcaption}
	\newtheorem{theorem}{Theorem}
	\newtheorem{corollary}{Corollary}
	\newtheorem{lemma}{Lemma}
	\newtheorem{remark}[theorem]{Remark}
	\newtheorem{definition}[theorem]{Definition}

\newcommand{\cE}{{\mathcal E}}

\newcommand{\cX}{{\mathcal X}}
\newcommand{\cY}{{\mathcal Y}}

\newcommand{\bN}{{\mathbb N}}

	\newcommand{\mkv}{-\!\!\!\!\minuso\!\!\!\!-} 

\newcommand{\typ}[1]{T_\delta^n(#1)}

\newcommand{\toas}[1]{\xrightarrow[#1]{}}

\providecommand{\abs}[1]{\left|#1\right|}

\begin{document}
\sloppy
	 
\title{Secrecy Capacity Region of Some \\ Classes of Wiretap Broadcast Channels} 

\author{
  \IEEEauthorblockN{Meryem~Benammar and Pablo~Piantanida}\\
  
\thanks{
The material in this paper was submitted in part to the IEEE Information Theory Workshop, Tasmania, Australia,  2-5 November 2014 and the 49th Annual Allerton Conference on Communications, Control, and Computing 2014. This work was accomplished when Meryem~Benammar was with the Dept. of Telecommunications at CentraleSupelec.}
\thanks{Meryem~Benammar is with the Mathematical and Algorithmic Sciences Lab, France Research Center, Huawei Technologies Co., Ltd (e-mail: meryem.benammar@huawei.com). }
\thanks{Pablo~Piantanida is with the Laboratoire des Signaux et Syst\`emes (L2S UMR 8506) at CentraleSupelec-CNRS-Universit\'e Paris-Sud, France (e-mail:  pablo.piantanida@centralesupelec.fr).}
\thanks{Copyright (c) 2014 IEEE. Personal use of this material is permitted.  However, permission to use this material for any other purposes must be obtained from the IEEE by sending a request to pubs-permissions@ieee.org.}}

\maketitle 
\begin{abstract}
This work investigates the secrecy capacity of the Wiretap Broadcast Channel (WBC) with an external eavesdropper where a source wishes to communicate two private messages over a Broadcast Channel (BC) while keeping them secret from the eavesdropper. We derive a non-trivial outer bound on the secrecy capacity region of this channel which, in absence of security constraints, reduces to the best known outer bound to the capacity of the standard BC. An inner bound is also derived which follows the behavior of both the best known inner bound for the BC and the Wiretap Channel. These bounds are shown to be tight for the deterministic BC with a general eavesdropper, the semi-deterministic BC with a more-noisy eavesdropper and the Wiretap BC where users exhibit a less-noisiness order between them. Finally,  by rewriting our outer bound to encompass the characteristics of parallel channels, we also derive the secrecy capacity region of the product of two inversely less-noisy BCs with a more-noisy eavesdropper. We illustrate our results by studying the impact of security constraints on the capacity of the WBC with binary erasure (BEC) and binary symmetric (BSC) components.
\end{abstract}

%================= SECTION =================

\section{Introduction} 
Information theoretic secrecy was first introduced by Shannon in his seminal work~\cite{Shannon1949}. He investigates a communication system between a source, a \emph{legitimate} receiver and an \emph{eavesdropper} where the source and the legitimate receiver share a secret key. It is shown that, to achieve perfect secrecy, one has to let the key rate be at least as large as the message rate. This result motivated the work~\cite{Wyner1975} by Wyner  who introduced the notion of Wiretap Channel. In such a setting, a source wishes to transmit a message to a \emph{legitimate} receiver in the presence of an \emph{eavesdropper} but without resorting to a shared key. Besides communicating reliably to the legitimate receiver at a maximum rate, the source has to maximize the equivocation at the eavesdropper so that it cannot recover the message sent over the channel. In the case of perfect secrecy, the conditional probability of the message given the eavesdropper's observation has to be approximately uniform over the set of messages, i.e., there is no leakage of information  to the eavesdropper. The surprising result of Wyner's work~\cite{Wyner1975} is that the use of a secret key is no longer required to guarantee a positive equivocation rate or even perfect secrecy. Csisz\'ar \& K\"orner's~\cite{CsiszarConfidential} generalized this result --first derived with the assumption of a degraded eavesdropper--  to the general BC and where the source must also transmit a common message to both users. As a matter of fact, an analysis of the corresponding rate region regarding the necessity of two auxiliary random variables, namely,  rate splitting and channel prefixing, was carried out by Ozel \& Ulukus in~\cite{Ozel2011}. It was shown that under specific channel ordering the rate region requires only one or even none of these variables.

Several multi-terminal Wiretap networks were studied, e.g., the MAC Wiretap Channel has been investigated by 
Liang \&  Poor in~\cite{Liang2006} while physical layer security in broadcast networks was studied by Liang {\it et al.} in~\cite{Liang-Poor-Shamai2009} though, the capacity region is yet to be fully characterized.

	\subsection*{Related works}
The Wiretap Broadcast Channel (WBC) was first studied under two types of secrecy constraints. The Broadcast Channel (BC) with confidential messages where the encoder transmits two private messages, each to its respective user, while keeping both of them secret from the opposite user. In~\cite{LiuMaric2008}, inner and outer bounds on the secrecy capacity were derived. The secrecy capacity of the semi-deterministic BC with confidential messages is derived in~\cite{Zhao2009} while in~\cite{Kang2009} it is assumed that only one message has to be kept secret from the other user and the capacity of the semi-deterministic eavesdropper setting was characterized. As for the Gaussian MIMO BC  with confidential messages, it was considered in the works of Liu {\it et al.} in~\cite{4787620, 5550390} while the Gaussian MIMO multi-receiver wiretap channel was addressed by Ekrem \& Ulukus in~\cite{5730586} (see~\cite{ITS-book,6582704} and references therein).

 An alternate setting is the BC with an \emph{external eavesdropper} where the secrecy requirement consists in that all messages be kept secret from the eavesdropper which is different from both users. Following this setting, the capacity of some classes of ordered and product BCs were first investigated by Ekrem \& Ulukus in~\cite{EkremUlukus} \cite{EU13}, where the legitimate users' channels exhibit a degradedness order and the eavesdropper is \emph{more-noisy} than all legitimate users' channels. In a concurrent work by Bagherikaram {\it et al.} in~\cite{KhandaniGaussian}, the secrecy capacity was characterized for the case where the eavesdropper is \emph{degraded} towards the weakest user and also for its corresponding additive white Gaussian noise (AWGN) channel model. 

\subsection*{Main contributions}
In this work, we consider the Wiretap BC where the encoder transmits two private messages to two users while it wishes to keep them secret from an external eavesdropper. We derive both an outer bound and an inner bound on the secrecy capacity region of this setting. The outer bound is obtained through a careful single-letter derivation that addresses the main difficulty of our setting which relies on upper bounding techniques for three terminals' problems. It should be emphasized that both converse techniques for the standard BC and the Wiretap Channel require the use of Csisz\'ar \& K\"orner's sum-identity~\cite{CsiszarConfidential} which does not apply to more than two output sequences. Besides this well-known difficulty, our outer bound clearly copies the mathematical form and behavior of the best known outer bound for the BC without an eavesdropper~\cite{NairELGamal2006}. As for the inner bound, our techniques simply follow the notion of \emph{double binning}, \emph{superposition coding} and \emph{bit recombination}. It also generalizes the inner bound of \cite{EU13} in the case of secure messages only, and, in the absence of secrecy requirement, the obtained inner bound naturally reduces to Marton's inner bound for the BC with common message~\cite{MartonInnerBound79}.
 	 
By developing an equivalent but non-straightforward representation of the outer bound, we show that it matches the inner bound for several novel classes of non-degraded Wiretap Broadcast Channels. More precisely, we are able to characterize the secrecy region of the following settings: 
\begin{enumerate}
\item The deterministic BC with an arbitrary eavesdropper where both legitimate users observe a deterministic function of the input, 
\item The semi-deterministic BC with a more-noisy eavesdropper where only one of the legitimate users is a deterministic channel while the other is less-noisy than the eavesdropper, 
\item The less-noisy BC with an eavesdropper degraded respect to the best legitimate user,   
\item The product of two inversely less-noisy BC with a more-noisy eavesdropper. 
\end{enumerate}
Besides of novel secrecy capacity results, the outer and the inner bound also recover some known results, e.g., the degraded BC with a more-noisy eavesdropper~\cite{EkremUlukus} which generalizes the degraded BC with a degraded eavesdropper~\cite{KhandaniGaussian}. 

We finally illustrate the results by investigating the impact of secrecy constraints on the capacity of the Wiretap Broadcast Channel with binary erasure (BEC) and binary symmetric (BSC) components. To his end, we derive the secrecy capacity region of a Less Noisy BEC/BSC BC with a degraded BSC eavesdropper and compare it to the standard capacity region, i.e. without secrecy constraints. In this setting, the central difficulty arises from the converse part for which we were able to show, through convexity arguments, a novel inequality on the conditional entropy of binary sequences. Indeed, this inequality appears to be crucial in the study of the WBC with BSC and BEC components, similar to Mrs. Gerber's lemma~\cite{Wyner1973} for the binary symmetric BC. The analysis of the secrecy capacity region proved that the degraded eavesdropper's impediment can be very severe on the BSC user whilst, it would still allow, for the worst degraded case, for positive rates for the BEC user.  
	
The remainder of this paper is organized as follows. In Section II, we give relevant definitions of the Wiretap BC setting and the main outer and inner bounds. We then show in Section III that the obtained bounds are tight for various classes of WBCs. In Section IV, we fully characterize the capacity region of the BEC/BSC Broadcast Channel with a BSC eavesdropper. Sections V, resp. VI are dedicated to the corresponding proofs of the outer, resp. inner bounds. Last, summary and concluding remarks are drawn in Section VIII.
 	
\subsection*{Notations}
For any sequence~$(x_i)_{i\in\mathbb{N}_+}$, notation $x_k^n$ stands for the collection $(x_k,x_{k+1},\dots, x_n)$. $x_1^n$ is simply denoted by $x^n$. Entropy is denoted by $H(\cdot)$, and mutual information by $I(\cdot;\cdot)$. $\mathds{E}$ resp. $\mathds{P}$ denote the expectation resp. the generic probability while the notation $P$ is specific to the probability of a random variable (rv). $\|\cX\|$ stands for the cardinality of the set $\cX$. We denote typical and conditional typical sets by $\typ{X}$ and $\typ{Y|x^n}$, respectively (see Appendix~\ref{sec:typical} for details). Let $X$, $Y$ and $Z$ be three random variables on some alphabets with probability distribution~$p$. If $p(x|yz)=p(x|y)$ for each $x,y,z$, then they form a Markov chain, which is denoted by $X\mkv Y\mkv Z$. The \emph{binary entropy function} $h_2$ is defined $\forall x \in [0:1]$ by 
$$ 
h_2(x)  \triangleq - x \log_2(x) - (1-x) \log_2 (1-x), 
$$ 
and the binary convolution operator $(\star)$ as: $x \star y \triangleq x  (1-y) + (1-x) y$ for all $(x ,y) \in [0:1]^2$ . 
		
We will use FME to designate \emph{Fourier-Motzkin} elimination. 
			
		%================= SECTION =================
		
	\section{System Model and secrecy capacity bounds}
Hereafter, we introduce the Wiretap Broadcast Channel (WBC) as represented in Fig.~\ref{WBC}, and then derive both an outer and an inner bound on its secrecy capacity region. 
	 
	\subsection{The Wiretap Broadcast Channel}
	\begin{figure}[t]
   \centering
   \includegraphics[scale=.6]{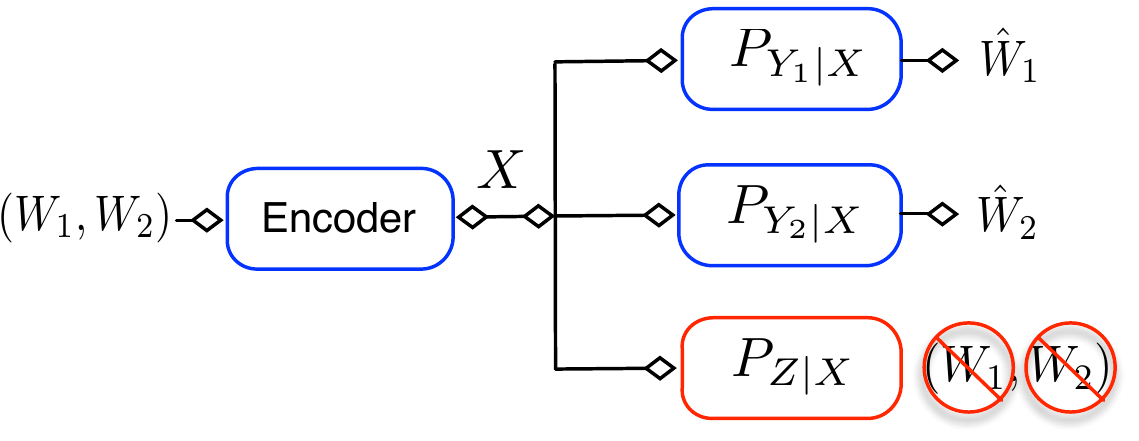} 
   \caption{The Wiretap Broadcast Channel (WBC).}
   \label{WBC} 
 \end{figure} 
	\begin{itemize}
	\item Consider an n-th extension of a three-user memoryless Broadcast Channel:
\begin{equation*}
	  \mathcal{W}^n  = \big\{P_{Y_1^{ n}Y_2^{ n}Z^{ n}|X^n }: \cX^n  \longmapsto \cY^n_1 \times \cY^n_2 \times \mathcal{Z}^n \big\}  \ , 
\end{equation*} 
defined by the conditional p.m.f: 
\begin{equation*}
  P_{Y_1^{n}Y_2^{n}Z^{n}|X^n }  \triangleq \prod^{n}\limits_{i=1} P_{Y_{1,i}Y_{2,i}Z_i|X_i}  \ .
\end{equation*} 
	\item An ${(M_{1n}, M_{2n}, n)}$-code for this channel consists of: two sets of messages $\mathcal{M}_1$ and $\mathcal{M}_2$, an encoding function that assigns an n-sequence $x^n (w_1, w_2) $ to each message pair $( w_1, w_2) \in  \mathcal{M}_1 \otimes \mathcal{M}_1 $ and decoding functions, one at each receiver, that assign to the received signal an estimate message $(\hat{w}_j)$ in $ \mathcal{M}_j ,  j\in \{ 1,2\} $ or an error.\\
	 The probability of error is given by: 
\begin{IEEEeqnarray*}{rCl}
 P_e^{(n)}  &\triangleq&  \mathds{P} \biggl( \bigcup_{j\in\{1,2\}} \bigl\{ \hat{W}_j \neq  W_j  \bigr\} \biggr). 
\end{IEEEeqnarray*}
	\item A rate pair $( R_1, R_2)$ is said to be achievable if there exists an ${(M_{1n}, M_{2n}, n)}$-code satisfying: 
\begin{IEEEeqnarray*}{rCl}
	  \liminf\limits_{n \rightarrow \infty} \frac{1}{n} \log_2 M_{jn} &\geq & R_j \,\,\,\textrm{ $\forall \,j \in \{1,2\}$}\,,\\
	  \limsup\limits_{n \rightarrow \infty}  \, P_e^{(n)} &=& 0 \ ,  \\
	  \liminf\limits_{n \rightarrow \infty} \frac{1}{n} H(W_1 W_2|Z^n) &\geq & R_1 + R_2 \ . 
\end{IEEEeqnarray*}
Note that the last constraint implies that for some sequence  $\epsilon_n$ of positive values:  
\begin{equation*}
I(W_1 W_2 ; Z^n ) \leq n \epsilon_n \ , 
\end{equation*}
which implies individual secrecy constraints given by
\begin{equation*} 
I( W_j ; Z^n ) \leq n \epsilon_n \ , \ \textrm{$\forall \,j \in \{ 1,2\}$}.\ 
\end{equation*}
\item The secrecy capacity region is the closure of the set of all achievable rate pairs $(R_1,R_2)$.
\end{itemize}
	\subsection{Ordered Broadcast Channels \cite{ElGamal02}}
	A Broadcast Channel $X \rightarrow (Y,Z)$ is said to be degraded, say $Z$ is degraded with respect to $Y$ if the following holds:
	\begin{equation*}
	\exists \ q(z|y) \ \  \text{such that} \  \  P(z|x) = \sum_{y\in \mathcal{Y}}  P(y|x)Q(z|y) \ , 
	\end{equation*}
	A channel output $Y$ is said to be ``less-noisy" than $Z$, or $Z$ is said to be ``more-noisy" than $Y$ if
	\begin{equation*}
	\forall P_U \ \text{ such that } \ U\mkv X\mkv (Y,Z) \ , \  I(U;Z) \leq I(U;Y) \ . 
	\end{equation*}
	\subsection{Outer bound on the secrecy capacity region of the WBC}
We next present an outer bound on the secrecy capacity region of the WBC under study. This bound originates from a careful single-letter characterization and accounts for different channel configurations which provides the secrecy capacity region for some new classes of wiretap broadcast channels.
\begin{theorem}[Outer bound]\label{TheoremOuter}
The secrecy capacity region of the Wiretap BC with an external eavesdropper is included in the set of rate pairs satisfying: 
x\begin{IEEEeqnarray}{rcl}\label{GeneralOuterBoundWBC}
 %\left\{
% \begin{array}{rcl}
     R_1 &\,\leq\,& 	I( U_1; Y_1 |T V_1 ) - I(U_1; Z  |TV_1)  \ ,\\
    R_1 &\leq& 	I( U_1; Y_1 Y_2 |T V_1  V_2) - I(U_1; Z  |TV_1 V_2) \ , \\ 
    R_1 &\leq& 	I( U_1; Y_1 |T V_1 U_2 ) - I(U_1; Z  |TV_1 U_2)  \ ,\\
    R_1 &\leq& 	I( U_1; Y_1 Y_2 |T V_1  U_2 V_2) - I(U_1; Z  |TV_1 U_2 V_2) \qquad   \\ 
	R_2 &\leq&  I( U_2; Y_2 |T V_2 ) - I(U_2; Z  |TV_2)  \ ,\\
	R_2 &\leq&  I( U_2; Y_2 Y_1 |T V_1 V_2 ) - I(U_2; Z  |T V_1 V_2)  \ ,\\
	R_2 &\leq&  I( U_2; Y_2 |T V_2 U_1) - I(U_2; Z  |TV_2U_1)  \ ,\\
	R_2 &\leq&  I( U_2; Y_2 Y_1 |TU_1 V_1 V_2 ) - I(U_2; Z  |T U_1 V_1 V_2)  \qquad   \\
	R_1 + R_2 &\leq& I(X; Y_2|T Z V_1) + I( U_1 S_1; Y_1 |T V_1 ) \nonumber  \\ 
	&& \qquad \qquad - I(U_1 S_1; Z Y_2 |T V_1) \ ,\\ 
	R_1 + R_2 &\leq& I(X; Y_2|T Z V_1 V_2) + I( U_1 S_1 ; Y_1 Y_2 |T V_1 V_2 ) \nonumber  \\ 
	&& \qquad \qquad - I(U_1 S_1 ; Z Y_2 |T V_1 V_2) \ ,\\
	R_1 + R_2 &\leq& I(X; Y_1|T Z V_2) + I( U_2 S_2; Y_2 |T V_2 )  \nonumber  \\ 
	&& \qquad \qquad- I(U_2 S_2; Z Y_1 |T V_2) \ , \\ 
\hspace{-2mm}	R_1 + R_2 &\leq& I(X; Y_1|T Z V_1 V_2) + I( U_2 S_2; Y_2 Y_1|T V_1 V_2 )  \nonumber  \\ 
	&& \qquad \qquad- I(U_2 S_2; Z Y_1 |T V_1 V_2) \ , 
% \end{array}% \right. 
\end{IEEEeqnarray}
	for some joint input p.m.f 
	$$
	P_{TV_1V_2 U_1 U_2 S_1 S_2X}=P_{T V_1V_2 U_1 U_2 S_1 S_2} P_{X| U_1  U_2  S_1 S_2}
	$$ 
	such that $ (T,V_1,V_2,S_1,S_2,U_1,U_2) \mkv X \mkv (Y_1,Y_2,Z)$.  
\end{theorem}
\begin{IEEEproof} The proof of this theorem is relegated to Section~\ref{ProofOuterBound}. \end{IEEEproof}

The next corollary proceeds to the reduction of some auxiliary rvs which can be removed without reducing the rate region. This simplifies the complexity of the optimization of the many variables present in the bound.
\begin{corollary}[Outer bound]\label{CorollaryOuter}
The rate region stated in Theorem~\ref{GeneralOuterBoundWBC} implies the next outer bound: 
\begin{IEEEeqnarray}{rcl}\label{OuterBoundWBC}	 
    R_1 &\,\leq\,& 	I( U_1; Y_1 | TV_1 ) - I(U_1; Z  |TV_1) \ , \\
    R_2 &\leq&  I( U_2; Y_2 |T V_2 ) - I(U_2; Z  |TV_2)  \ ,\\
	R_1 + R_2 &\leq& I(X; Y_2| TZ  V_1   ) + I( U_1 ; Y_1 | TV_1 ) \nonumber \\
	&& \qquad \qquad - I(U_1 ; Z Y_2 |TV_1)	\ , \\
	R_1 + R_2 &\leq& I(X; Y_1|T Z V_2) + I( U_2 ; Y_2 |T V_2 ) \nonumber \\
	&& \qquad \qquad - I(U_2 ; Z Y_1 |T V_2) \ ,
	\end{IEEEeqnarray} 
	for some joint input p.m.f $P_{TV_1V_2 U_1 U_2X} $ such that $ (T,V_1,V_2,U_1,U_2) \mkv X \mkv (Y_1,Y_2,Z)$. 
	\end{corollary}
	\begin{IEEEproof} The proof is relegated to Section~\ref{EliminationS}. \end{IEEEproof}
It is easy to check that by removing the secrecy constraint, i.e., if $Z$ is dropped, the above rate region reduces to the best known outer bound to the capacity of the standard BC~\cite[Lemma 3.5]{NairELGamal2006}. Moreover, this outer bound will prove to be crucial to characterize the secrecy capacity of several classes of WBCs, as will be stated later on.  
	
	\subsection{Inner bound on the secrecy capacity region of the WBC}
	In this section, we present an inner bound on the secrecy capacity region of the WBC. The coding argument combines both stochastic encoding to achieve secrecy and the standard coding techniques for the BC, i.e.,  superposition coding and random binning to let the sent codewords be arbitrarily dependent. 
	
	\begin{theorem}[Inner bound]\label{TheoremInner}
	The secrecy capacity region of the WBC includes all rate pairs $(R_1,R_2)$ satisfying: 
	\begin{IEEEeqnarray}{rcl}\label{InnerBoundWBC}
	%	 \left\{\begin{array}{rcl}
	    R_1 &\,\leq\,& I( Q U_1; Y_1 | T ) - I( Q U_1; Z | T )  \ ,\\
		R_2 &\leq& I( Q U_2; Y_2 | T ) - I( Q U_2; Z | T )  \ ,\\
		R_1 + R_2 &\leq& I( U_1; Y_1| T Q ) + I( Q U_2; Y_2 | T ) \nonumber \\
	&& \qquad- I(Q U_1 U_2 ; Z | T ) - I(U_1; U_2 | T Q ) \ ,\\
		R_1 + R_2 &\leq& I( U_2; Y_2| T Q ) + I( Q U_1; Y_1 | T ) \nonumber \\
	&& \qquad  - I(Q U_1 U_2 ; Z | T ) - I(U_1; U_2 | T Q ) \ ,\\
		R_1 + R_2 &\leq& I( Q U_1; Y_1| T ) + I( Q U_2; Y_2 | T )- I(Q U_1 U_2 ; Z | T ) \nonumber \\
	&& \qquad  - I(U_1; U_2 | T Q ) - I(Q ; Z | T) \ , 
		%\end{array}\right.  
		\end{IEEEeqnarray}
	for some joint p.m.f $P_{T Q U_1 U_2 X}$ such that $ (T,Q,U_1,U_2) \mkv X \mkv (Y_1,Y_2,Z)$ and $I( U_2; Y_2| T Q ) + I( U_1; Y_1 | Q T ) \geq  I(U_1; U_2 | T Q ) $.
	\end{theorem}
\begin{IEEEproof} The full proof of this inner bound is given in Section~\ref{ProofInnerBound}. \end{IEEEproof} 
\begin{remark}
It is worth mentioning here the relative behavior of this inner bound with the one of Theorem 1 in \cite{EU13} where the authors relied on similar encoding techniques as the ones we resort to in the proof of achievability. 

The corresponding inner bound is clearly included in \label{InnerBoundWBC} and it can be investigated whether these two inner bounds are indeed equal, similarly to \cite{LiangKramerEquivalence2008}, since the encoding is similar and only decoding strategies differ: successive decoding for \cite{EU13} and joint decoding in our case. 
\end{remark}

%================= SECTION =================

	\section{Secrecy Capacity of Some Wiretap Broadcast Channels}
	In this section, we derive the secrecy capacity of various Wiretap Broadcast Channel models.
\subsection{Deterministic BC with an arbitrary eavesdropper}
	Let us assume that both legitimate users' channel outputs are deterministic functions of the input $X$, as shown in Fig.~\ref{FigDeterministic}. 
	\begin{figure}[t]
   \centering
   \includegraphics[scale=.6]{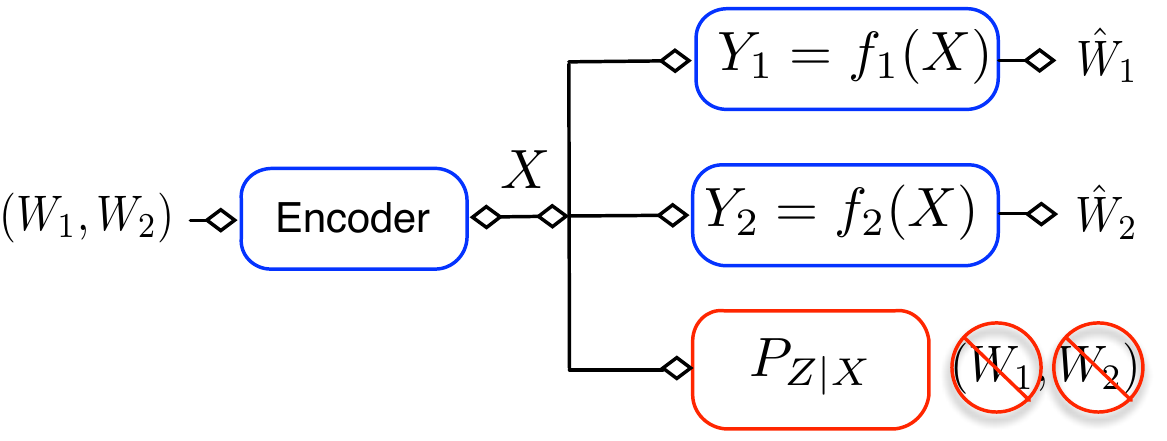} 
   \caption{Deterministic BC with an arbitrary  eavesdropper.}  \label{FigDeterministic}
\end{figure}
	\begin{theorem}[Secrecy capacity of the deterministic BC with a general eavesdropper] 
	The secrecy capacity of the deterministic BC with an arbitrary eavesdropper's channel is given by the set of all rate pairs $(R_1, R_2)$ satisfying: 
	\begin{IEEEeqnarray}{rcl}
	R_1 &\,\leq\,& H(Y_1| Z ) \ , \\
	R_2 &\leq& H(Y_2|Z) \ , \\
	R_1 + R_2 &\leq& H(Y_1  Y_2|Z)   \ ,
	\end{IEEEeqnarray}
	for some input p.m.f $P_X$. 
	\end{theorem}
	\begin{proof}
	We start with the achievability part for which we evaluate the inner bound in Theorem~\ref{TheoremInner} by setting: $Q = \emptyset$, $U_1 = Y_1$ and $U_2 = Y_2$. The claim follows then in a straightforward manner. As for the outer bound, it follows from the reduced outer bound in Corollary~\ref{CorollaryOuter}, by writing the next set of inequalities for $j \in \{1,2\}$: 	
	\begin{IEEEeqnarray}{rCl}
	  I(U_j; Y_j | V_j) &-&  I(U_j; Z| V_j)  \nonumber \\
	    &\leq& I(U_j; Y_j Z| V_j) - I(U_j; Z| V_j) \\ \label{Inequ_1_d}
	    &= &I(U_j; Y_j |Z,V_j) \\ 
	    &\leq&  H(Y_j |Z ) \label{Inequ_1_f}  
	\end{IEEEeqnarray}	
	with strict equality if $U_j = Y_j$ and $ V_j = \emptyset$. 
	Note also that: 
	\begin{IEEEeqnarray}{rCl}
	I(X ; Y_2 | Z V_1) &+& I(U_1; Y_1| V_1) - I(U_1; Z Y_2| V_1)   \nonumber \\
	   &\leq&  I(X ; Y_2 | Z V_1) + I(U_1; Y_1|Z Y_2 V_1) \\
	   &\leq&  H( Y_2 | Z V_1) + H(Y_1|Z Y_2 V_1) \\
	   &\leq& H( Y_1Y_2| Z ) \label{Inequ_2_f} 
	\end{IEEEeqnarray}	
	with strict equality if $U_1 = Y_1$ and $ V_1 = \emptyset$. The second sum-rate yields the same constraint.
	Thus, the outer bound is maximized with the choice $U_1 = Y_1$, $U_2 = Y_2$ and $V_1 = V_2 =\emptyset$. 
	\end{proof}
	
Below, we generalize the equality between the regions in Corollary~\ref{CorollaryOuter} and Theorem~\ref{TheoremInner} to the case of the Semi-Deterministic BC with a more-noisy eavesdropper. 

%% -------------------------------------------------------------------------------------------------------------------%%
	\subsection{Semi-deterministic BC with a more-noisy eavesdropper}
	Let us assume that only $Y_1$ is a deterministic function of $X$ but we further assume that $Y_2$ is less-noisy respect to the eavesdropper's output $Z$,  as shown in Fig.~\ref{FigSemiDeterministic}.
	\begin{figure}[t]
   \centering
   \includegraphics[scale=.6]{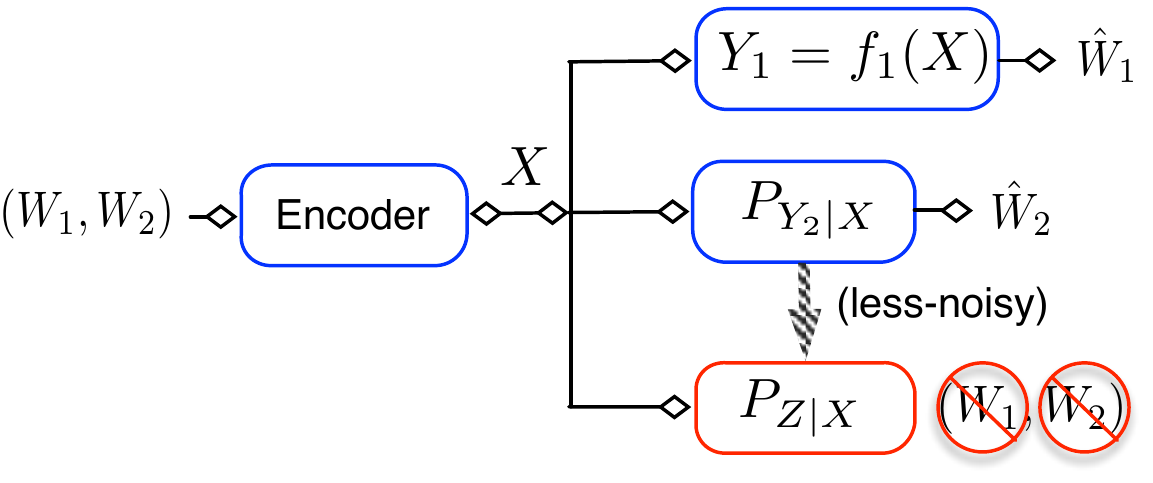} 
   \caption{The Semi-deterministic Wiretap Broadcast Channel with a more-noisy eavesdropper.}  \label{FigSemiDeterministic}
\end{figure}
	\begin{theorem}[Secrecy capacity region of the semi-deterministic BC with a more-noisy eavesdropper]\label{SemiDeterministicCapa}
	The secrecy capacity of the semi-deterministic BC with a more-noisy eavesdropper is the set of all rate pairs $(R_1, R_2)$ satisfying: 
	\begin{IEEEeqnarray}{rcl}
	%\left\{\begin{array}{rcl}
	R_1 &\,\leq\,& H(Y_1|Z Q) \ , \\
	R_2 &\leq& I(U;Y_2|Q) - I(U;Z|Q)  \ ,\\
	R_1 + R_2 &\leq& H(Y_1|ZQ U) + I(U;Y_2|Q) - I(U;Z|Q) \ 
	%\end{array}\right. 
	\end{IEEEeqnarray}
	for some joint p.m.f $P_{QUX}= P_{Q}P_{U|Q}P_{X|U} $ such that $(Q,U) \mkv X \mkv (Y_1, Y_2,Z)$\ .  
	\end{theorem}
	\begin{proof}
	Achievability follows from the rate region stated in Theorem~\ref{TheoremInner} by letting: $Q = T$ and $U_1 = Y_1$. As for the converse, we will first evaluate the outer bound given in Corollary~\ref{CorollaryOuter}. Since $Y_2$ is less-noisy than $Z$, then one can easily notice that:
\begin{equation} 
 I(U;Y_2|V Q) - I(U;Z| VQ )   \leq  I(UV;Y_2| Q) - I(UV;Z| Q )  \ . 
\end{equation}	
Considering the same chain of inequalities as in~\eqref{Inequ_1_d}-\eqref{Inequ_1_f}, one can write the outer bound as: 
\begin{IEEEeqnarray}{rcl}
	R_1 &\,\leq \,& H(Y_1|Z Q)  \ ,\\
	R_2 &\leq& I(U V;Y_2| Q) - I(UV;Z| Q ) \ , \\
	\hspace{-5mm} R_1 + R_2 &\leq&  H(Y_1|ZQUV) + I(UV;Y_2 | Q) - I(UV;Z|Q) \  
	\end{IEEEeqnarray}
	and thus, defining $(UV) = U $, we can write that the outer bound is the union over all p.m.f $P_{QUX}= P_{Q}P_{U|Q}P_{X|U} $ of the rate region given in Theorem~\ref{SemiDeterministicCapa}. 
	\end{proof}
	\begin{remark}
	When $Y_2$ is not less-noisy than $Z$, it is not clear yet whether the two bounds can be tight due to the fact that the auxiliary rv $V$ does not seem to be useless then. 
	\end{remark}
	
	\subsection{Degraded BC with a more-noisy eavesdropper}
	In this section, we assume that the legitimate user $Y_2$ is degraded respect to the legitimate user $Y_1$. Moreover, assume that both users are less-noisy than the eavesdropper as shown in Fig.~\ref{Ordered1}. 
	\begin{figure}[t]
   \centering
   \includegraphics[scale=.6]{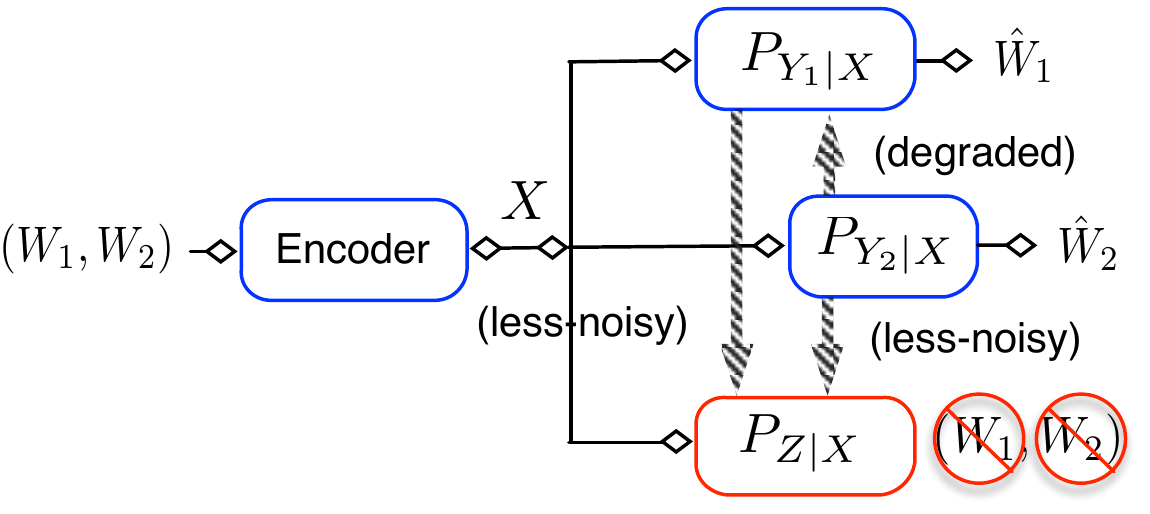} 
   \caption{Degraded BC with a more-noisy eavesdropper.}  
   \label{Ordered1}
\end{figure} 
The capacity region of this setting was first derived in~\cite{Ulukus}, and here, we simply rely on the optimality of our outer bound for this setting.  
 	 	\begin{theorem}[Secrecy capacity region of the degraded WBC \cite{Ulukus}] \label{CapaDegraded}
 	 	The secrecy capacity region of the degraded WBC is given by the set of rate pairs $(R_1, R_2)$ satisfying: 
 	 		 \begin{IEEEeqnarray}{rcl} 
    R_1 &\, \leq\, & I( X ; Y_1 |T U ) - I(X ; Z |T U)  \ , \\ 
	R_2 &\leq& I( U ; Y_2 |T   ) - I(U ; Z |T  )  \ ,   
	\end{IEEEeqnarray}
	for some input p.m.f $P_{TUX}$ where $(T,U)\mkv X \mkv (Y_1,Y_2, Z)$. 
	\end{theorem}
	
	\begin{proof}
	To show this, we first note that the outer bound given in Theorem~\ref{TheoremOuter} is included in the following outer bound obtained through keeping only the constraints: 
	\begin{IEEEeqnarray}{rcl} 
    R_1 &\,\leq\, & 	I( U_1; Y_1 Y_2 |T V_1 U_2 V_2) - I(U_1; Z  |TV_1U_2 V_2) \ , \\ 
	R_2 &\leq&  I( U_2; Y_2 |T V_2 ) - I(U_2; Z |TV_2)  \ .   
	\end{IEEEeqnarray}
	Now, since $Y_2$ is degraded respect to $Y_1$, then
	\begin{equation}
	  I( U_1; Y_1 Y_2 |T V_1 U_2 V_2) = I( U_1; Y_1 |T V_1 U_2 V_2)  \ , 
	 \end{equation} 
	 and since $Y_1$ is less-noisy than $Z$  we can write 
	 \begin{IEEEeqnarray}{rCl}
	 && I( U_1; Y_1 |T V_1 U_2 V_2) - I(U_1; Z  |TV_1 U_2 V_2)\nonumber \\
	  &\leq& I( U_1V_1; Y_1 |T V_1 U_2  ) - I(U_1 V_1; Z  |T U_2 V_2) \\
	 &\leq& I( X; Y_1 |TU_2 V_2) - I(X; Z  |T U_2 V_2) \ .
	 \end{IEEEeqnarray}
Thus, the outer bound reduces to the union over all joint p.m.fs $P_{TUX}$ of the rate region given in Theorem~\ref{CapaDegraded}. 
\end{proof}
	 In the sequel, it turns out that the outer bound we derived yields also the capacity region of another class of ordered BC, which does not include the class of degraded BC with a more-noisy eavesdropper as will be clarified shortly. 
%% -------------------------------------------------------------------------------------------------------------------%%
\begin{figure}[t]
   \centering
   \includegraphics[scale=.6]{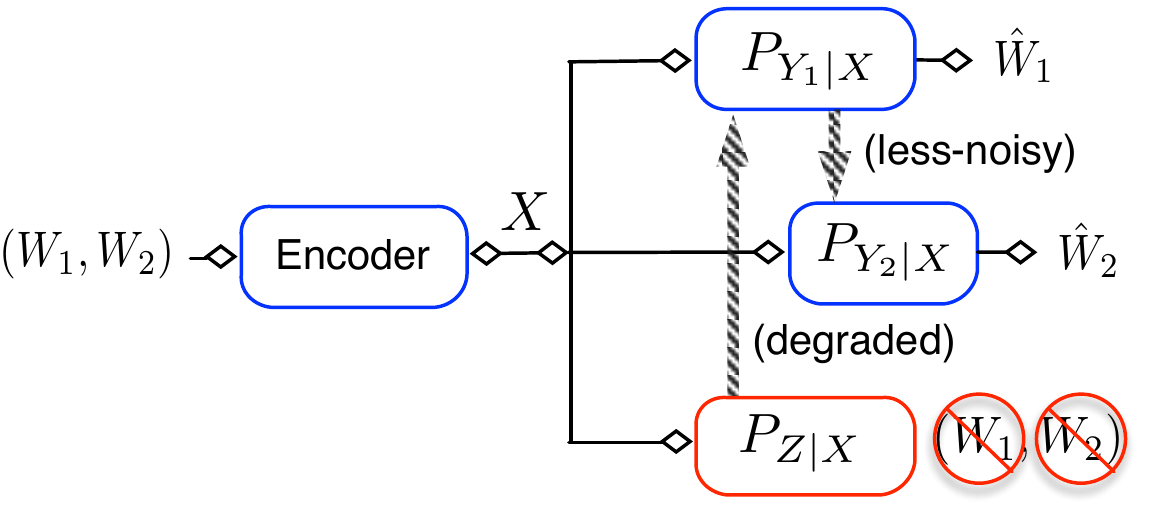} 
   \caption{Less-Noisy BC with a partly degraded eavesdropper.}  \label{Ordered2}
\end{figure}

\subsection{Less-Noisy BC with a partly degraded eavesdropper}
	Let us assume that $Y_1$ is a less-noisy channel than $Y_2$ and that $Z$ is a degraded version of $Y_1$. 	As shown in Fig.~\ref{Ordered2}, this model is more general than the one first considered in~\cite{KhandaniGaussian}, while it does not really generalize the model in Fig.~\ref{Ordered1}, first considered in~\cite{Ulukus}. Notice that in this setting the eavesdropper is not compulsorily degraded. However, the present class is wider in that users are no longer compulsorily degraded between them and the eavesdropper is no longer more noisy than the weaker legitimate user. 
	
	\begin{theorem}[Secrecy capacity region of the less-noisy WBC]\label{OrderedWBC}
	The secrecy capacity region of the ordered WBC under study is the set of all rate pairs $(R_1, R_2)$ satisfying: 
	\begin{IEEEeqnarray}{rcl}
	R_2 &\, \leq \, & I(U;Y_2| T) - I(U;Z| T)  \ ,\\   
	R_1 + R_2 &\leq& I(X;Y_1|Z U T) + I(U;Y_2|T) - I(U;Z|T)  \ ,\,\,\,\,\,\,
	\end{IEEEeqnarray}
	for some joint p.m.f $P_{TUX}= P_{T}P_{U|T}P_{X|U}$ such that $(T,U) \mkv X \mkv (Y_1,Y_2,Z)$.  
	\end{theorem}
\begin{proof} 
	The converse follows from the outer bound in Corollary~\ref{CorollaryOuter} by keeping only the terms: 
	\begin{IEEEeqnarray}{rcl}
    R_2 &\, \leq\, &  I( U_2; Y_2 |T V_2 ) - I(U_2; Z | TV_2)  \ ,\\
	R_1 + R_2 &\leq& I(X; Y_1|T Z U_2 V_2) + I( U_2 ; Y_2 |T V_2 )   \nonumber \\ 
	&& \qquad \qquad - I(U_2 ; Z Y_1 |T V_2)   \ ,  
	\end{IEEEeqnarray}
	and defining the common auxiliary rv $T \equiv (T,V_2)$. As for the achievability, let $U_1 = X$ and $Q = U_2$ in the inner bound given by Theorem~\ref{TheoremInner}. This bound reduces to:
	\begin{IEEEeqnarray}{rcl}
	R_1 &\, \leq\, & I(X;Y_1|T) - I(X;Z|T) = I(X;Y_1|ZT)  \ , \,\,\,\\
	R_2 &\leq& I(U;Y_2|T) - I(U;Z|T) \ ,\,\,\,\\  
	R_1 + R_2 &\leq& I(X;Y_1|Z U T ) + I(U;Y_2|T) - I(U;Z|T) \ , \,\,\,\,\,\,\label{eq-especial1}\\
	R_1 + R_2 &\leq& I(X;Y_1|T) - I(X;Z|T) = I(X;Y_1| Z T) \ .\,\,\,\label{eq-especial2}
	\end{IEEEeqnarray}

The first bound is redundant with respect to the last one. Moreover, since $Y_1$ is less-noisy than $Y_2$, then the bound~\eqref{eq-especial2} becomes redundant with respect to~\eqref{eq-especial1}. The inner bound reduces henceforth to the one given in Theorem~\ref{OrderedWBC}.  
\end{proof}
In the sequel, we study a non-straightforward extension of this WBC for which the secrecy capacity region remained  open since the previous results in literature apply only to the degraded BC case.	

%------------------------------------------------------------------------------------------------------------% 	
   \subsection{Product of two inversely less-noisy wiretap broadcast channels}
The product of inversely less-noisy broadcast channels is defined as the product of two less-noisy WBCs. The BC $(Y_1,T_1)$ has a component $Y_1$ which is less-noisy than $T_1$ and an eavesdropper $Z_1$ is degraded towards the best user $Y_1$ and more-noisy than the worst user $T_1$. The BC $(Y_2, T_2)$ is less-noisy in the inverse order and the eavesdropper $Z_2$ is degraded towards $T_2$ and more-noisy than $Y_2$.   
 \begin{figure}[t]
   \centering
   \includegraphics[scale=.6]{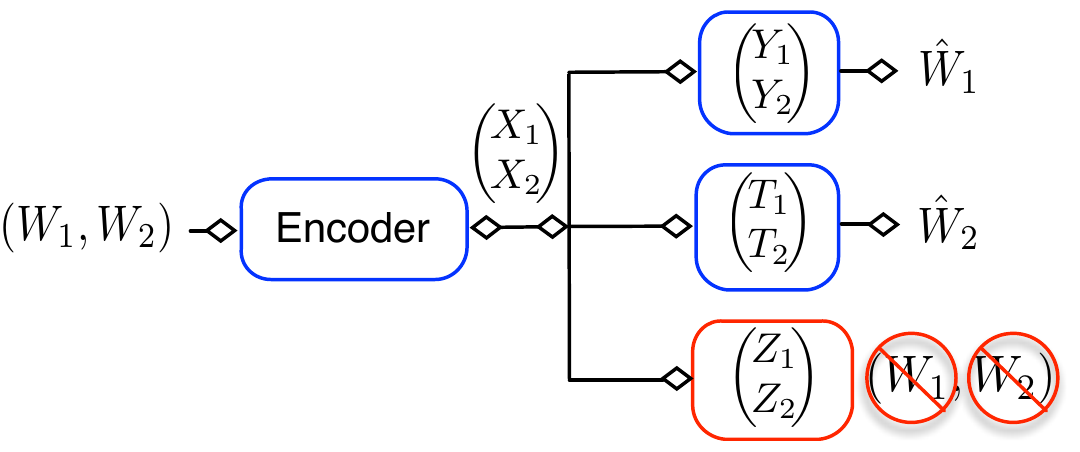} 
   \caption{The Parallel Broadcast Channel (PBC) with an eavesdropper.}  \label{FigSystemModel}
\end{figure} 

 \begin{theorem}[Product of two inversely less-noisy BCs with a more-noisy eavesdropper]\label{ProductInverselyLN}
 The secrecy capacity region of such a setting is given by the set of rates pairs $(R_1,R_2)$ satisfying: 
	\begin{IEEEeqnarray}{rcl}
   R_1 &\, \leq \, & I(X_1; Y_1|Z_1) + I(U_2; Y_2) - I(U_2;Z_2) \ , \\
R_2 &\leq & I(X_2; T_2|Z_2) + I(U_1; T_1) - I(U_1;Z_1) \ , \\
R_1 + R_2 &\leq& I(X_1; Y_1|Z_1) + I(U_2; Y_2) - I(U_2;Z_2)  \nonumber \\
&& \qquad \qquad  + I(X_2;T_2 |Z_2 U_2) \ , \\
R_1 + R_2 &\leq& I(X_2; T_2|Z_2) + I(U_1; T_1) - I(U_1;Z_1) \nonumber \\
&& \qquad \qquad  +I(X_1; Y_1|Z_1 U_1) \ , 
 \end{IEEEeqnarray}
for some input p.m.f $P_{U_1X_1U_2X_2} = P_{U_1X_1}P_{U_2X_2} $ that satisfies  $(U_1,U_2) \mkv (X_1,X_2) \mkv (Y_1,Y_2,T_1,T_2,Z_1,Z_2)$.
 \end{theorem}
\begin{proof}
The proof is quite evolved in that it requires a new outer bound formulation, and is thus relegated to Appendix~\ref{ProofProduct}. 
\end{proof}
Note here that, in the absence of the eavesdropper, this theorem yields the capacity region of the product of two reversely less-noisy BCs which, though not proved in \cite{ElGamalProductSum}, can be deducted from the result of \cite{Geng2011} for the product of reversely more-capable BCs. 

%================= SECTION =================

\section{The BEC/BSC Broadcast Channel with a BSC eavesdropper }
In this section, we characterize the capacity region of the BEC/BSC broadcast channel with an external BSC eavesdropper. This model falls into the class of ordered BCs and is extremely rich since the BC (BEC and BSC) provides for a variety of orderings following the respective values of the erasure probability ``$e$" and the crossover probability ``$p$", as it is summarized in the table~\ref{TableBECBSC} and shown in~\cite{NairBECBSC2010}. Let us consider the channel model where: 
\begin{equation}
\mathcal{W}\,:\,\left\{
\begin{array}{lcl}
\cX &  \longmapsto & \cY_1 \equiv \textrm{BEC($e$)}\, , \\
\cX &  \longmapsto & \cY_2 \equiv \textrm{BSC($p_2$)}\, , \\ 
\cX &  \longmapsto & \mathcal{Z} \,\, \equiv \textrm{BSC(p)}\, .  
\end{array}\right. \vspace{1mm}
\end{equation}

\begin{table}[ht!]
\centering
 \begin{tabular}{|c|c|c|c|c|}
\hline 
 $\!\!\! 0 \leq e \leq 2p \!\!\!$ &$\!\!\! 2p < e \leq 4p(1-p) \!\!\!$& $\!\!\! 4p(1-p) < e \leq h(p)\!\!\! $&$\!\!\! h(p) < e \leq 1 \!\!\!$ \\ 
\hline 
  Degraded & Less-noisy & More-capable & Es.Less-noisy  \\
\hline 
\end{tabular}
\caption{\label{TableBECBSC} Different orderings allowed by BEC(e) and BSC(p) models.}
\end{table}

We will consider the case where $Y_1$ is less-noisy that $Y_2$ and where $Z$ is degraded towards $Y_2$. Besides, we make sure that $Z$ is degraded towards $Y_1$. \footnote{It is worth emphasizing here that our choice of $Z$ degraded respect to $Y_2$ follows from that both channels are naturally degraded since these are $BSC$ channels. Otherwise, if $Y_2$ were to be degraded respect to $Z$, no positive rate could  be transmitted to user $2$ .} Summarizing these constraints, we end up with the inequalities: 
\begin{equation}\label{ConstraintBECBSC}
2 p_2 \leq e \leq \min \{2 p, 4 p_2 (1 - p_2)\} \ . 
\end{equation}

\begin{theorem}[Secrecy capacity region of the BEC($e$) /BSC($p_2$) BC with BSC($p$) eavesdropper]
The capacity region of the BC with BEC($e$) / BSC($p_2$)  components and a BSC($p$) eavesdropper, defined by the constraint~\eqref{ConstraintBECBSC} where $1 - 4 p (1 - p) \geq 4 p_2 (1 - p_2)$, is given by the set of rate pairs satisfying: 
\begin{equation}\label{CapacityBECBSC}
\mathcal{C}  \,: \left\{ \begin{array}{rcl}
R_1 &\leq& (1-e) \, h_2(x) + h_2(p) - h_2(p\star x) \ , \\
R_2 &\leq&  h_2(p\star x) - h_2(p_2\star x) \ ,
\end{array}\right.  
\end{equation}
for some $x \in [0:0.5]$.  
\end{theorem}

\begin{proof}
The proof consists in evaluating the capacity region of such an ordered channel given by $\mathcal{R}$, the set of rate pairs $(R_1,R_2)$ satisfying: 
\begin{equation}\label{SingleLetter}  \begin{array}{rcl}
R_1 &\leq& I(X;Y_1 |T U) - I(X;Z|TU) = I(X;Y_1|Z T U) \ , \\
R_2 &\leq& I(U;Y_2|T) - I(U;Z|T) = I(U;Y_2 |Z T) \ , 
\end{array}  
\end{equation} 
and is two fold. The challenging part is obviously the converse part since it requires the use of an inequality, similar in a way to Mrs. Gerber's lemma~\cite{Wyner1973} applied to the secrecy capacity region, which we have been able to prove only under the assumption $1 - 4 p (1 - p) \geq 4 p_2 (1 - p_2)$, although there is strong evidence that the converse can be proved besides this case.  

Note that $T = \emptyset$ maximizes the region since it can easily be shown to be convex and thus, will not need the time-sharing variable $T$. Moreover, we can state a cardinality bound on the auxiliary rv $U$ used in evaluating the previous region following the usual Fenchel-Eggleston-Caratheodory theorem that is it suffices to evaluate the region using an auxiliary rv with a quaternary alphabet.  

First, note that the choice $X = U \oplus V$ where $U \sim \textrm{Bern} (0.5)$, $V \sim \textrm{Bern} (x)$ yields that $X \sim \textrm{Bern} (0.5)$ and that $X|U \sim \textrm{Bern}(x)$. Thus, we can write: 
	\begin{IEEEeqnarray}{rcl}
 I(X;Y_1 |U)&\, = \, & (1-e) H(X|U) = (1-e) h_2(x) \ ,\\
 I(X;Z|U) &=& h_2(p * x) - h_2(p)  \ ,\\
 I(U;Y_2) &=& 1 - h_2(p_2* x) \ , \\
 I(U;Z) &=& 1 - h_2(p*x) \ ,
\end{IEEEeqnarray}
which proves the inclusion of the region $ \mathcal{R} $ in the rate region $\mathcal{C}$, i.e.,  the achievability.  

As for the inclusion in the appositive way, i.e.,  the converse, we will use the following lemma. 
\begin{lemma}\label{LemmaConvexity}
If $1 - 4p(1-p) \leq 4 p_2 (1-p_2)$, then $\mathcal{R}$ defines a convex set. 
\end{lemma}
\begin{IEEEproof}  The proof is given in Appendix~\ref{ProofLemmaConvexity}. \end{IEEEproof}
 
Now, since $\mathcal{R}$ and $\mathcal{C}$ define convex bounded sets, then both are uniquely defined by their supporting hyperplanes. 
And finally, since $\mathcal{R}$ is included in $\mathcal{C}$, it thus suffices to show that all their supporting hyperplanes intersect, so let then $\lambda \in [0: \infty[$. We want to show that\footnote{Note that the maxima are well defined for both regions due to the cardinality bound (for $\mathcal{C}$) and for the closed and bounded interval for $\mathcal{R}$ which results in compact supports for both optimizations.}: 
\begin{equation}
\underset{(R_1,R_2) \in \mathcal{C} }{\max}  R_1 + \lambda R_2 \leq  \underset{(R_1,R_2) \in \mathcal{R} }{\max}  R_1 + \lambda R_2  \ .
\end{equation}
 
Let us choose the following notation: $U$ is an auxiliary rv that takes its values in $\mathcal{U} = \{1 , \dots ,  \|\mathcal{U }\| \}$ following the law: $\mathds{P}(U = u) = P_U(u) \triangleq P_u$. Let us assume that $X$ is a $\textrm{Bern}(\alpha)$ distributed Binary rv and that\footnote{$\mathcal{U}$ is the support of the law $P_U$, as such, $P_{X|U}(0|u)$ is well defined.} $\mathds{P}(X= 0|U = u) =P_{X|U}(0|u)  \triangleq x_u$.

Define the set $\mathcal{P}$ of admissible transition probabilities as: 
\begin{IEEEeqnarray}{rCl}
&&  \mathcal{P}  \triangleq \biggl\{  (\alpha , \textbf{x}_{\|\mathcal{U}\|} ,  \textbf{p}_{\|\mathcal{U}\|} )  =  (\alpha , x_1,\dots,x_{\|\mathcal{U}\|} , p_1, \dots ,p_{\|\mathcal{U}\|} )    \nonumber \\ 
  && \qquad  \qquad   \qquad  \qquad  \in   \left[0:0.5\right]^{\|\mathcal{U}\| +1 } \times \left[0:1\right]^{\|\mathcal{U}\|}     \nonumber \\ 
  &&  \qquad  \qquad \text{ s.t }   \sum_{u=1}^{\|\mathcal{U}\|} p_u = 1 \ , \  \sum_{u=1}^{\|\mathcal{U}\|} p_u x_u  = \alpha   \biggr\}  \ . 
\end{IEEEeqnarray}	
With this, note that:  
\begin{IEEEeqnarray}{rCl}
&& 	   \underset{(R_1,R_2) \in \mathcal{C} }{\max}  R_1 + \lambda R_2 \nonumber \\
 &=& \underset{ \substack{ P_{UX} \\ U  \mkv  X  \mkv  (Y_1,Y_2,Z)} }{\max} I(X;Y_1|U) - I(X;Z|U)    \nonumber \\ 
  && \qquad \qquad \qquad  + \lambda \Bigl[ I(U;Y_2) - I(U;Z) \Bigr] \\
	   &=&  \underset{   (\alpha \,, \,\textbf{x}_{\|\mathcal{U}\|} \,, \,\textbf{p}_{\|\mathcal{U}\|} ) \in \mathcal{P} }  {\max} h_2(p) + \lambda \Bigl[ h_2(p_2*\alpha) - h_2(p*\alpha) \Bigr]  \nonumber \\
	     && \quad + \sum_{u\in \mathcal{U}} P_u  \left\{ \vphantom{ \Bigr]}(1-e) h_2(x_u) - h_2(p*x_u) \right.   \nonumber \\ 
  && \qquad  \left. +  \lambda \Bigl[  h_2(p*x_u) -  h_2(p_2*x_u) \Bigr] \right\} \\
	    &\overset{(a)}{\leq}&  \underset{   (\alpha\, , \,\textbf{x}_{\|\mathcal{U}\|} \,, \,\textbf{p}_{\|\mathcal{U}\|} ) \in \mathcal{P} }  {\max} \quad h_2(p) \nonumber \\
	     && \quad + \sum_{u\in \mathcal{U}} P_u  \left\{ \vphantom{\Bigr]} (1-e) h_2(x_u) - h_2(p*x_u) \right.   \nonumber \\ 
  && \qquad  \left. +  \lambda \Bigl[  h_2(p*x_u) -  h_2(p_2*x_u) \Bigr] \right\} \ \ \ \ \\ 
	    &\overset{(b)}{\leq}&  h_2(p) +   (1-e) h_2(x_u^{\lambda}) - h_2(p*x_u^{\lambda})   \nonumber \\ 
  && \qquad   +  \lambda \Bigl[  h_2(p*x_u^{\lambda}) -  h_2(p_2*x_u^{\lambda}) \Bigr]  \\
	    &=& \underset{(R_1,R_2) \in \mathcal{R} }{\max}  R_1 + \lambda R_2  \ , 
	\end{IEEEeqnarray}	 
	where:
	 \begin{IEEEeqnarray}{rCl}
	x_u^{\lambda} &=& \arg \max \left\{ \vphantom{\Bigr]} (1- e) h_2(x) - h_2(p*x) \right.   \nonumber \\ 
  && \qquad  \left. \qquad  +  \lambda \Bigl[  h_2(p*x) -  h_2(p_2*x) \Bigr] \right\} \ .  
	\end{IEEEeqnarray}
Now, $(a)$ follows from the fact that since $x , p_1 , p_2 \in [0:1/2]$ and $p \geq p_2$, then: 
\begin{IEEEeqnarray}{rCl}
  &&\forall  \alpha \in [0:1/2]   \quad  ,   \quad   p_2 * \alpha \leq p * \alpha  \leq 1/2\\ 
  \textrm{ then }     && \max_{\alpha \in [0:1/2]}  \left[h_2( p_2 * \alpha  ) - h_2( p * \alpha  )\right] = 0  
\end{IEEEeqnarray}	
with equality for $\alpha = 1/2$. As for $(b)$, it is a direct result of the existence of a value of $x_u^{\lambda}$ that maximizes the expression, and from that letting $\mathcal{U} = \{0,1\}$ and $P_0 = P_1 = \frac{1}{2}$ and $U \longmapsto  X \equiv BSC(x_u^{\lambda}) $, leads to this maximum value equality in $(b)$ in addition to being admissible: $P_0 x_u^{\lambda} + P_1 (1-x_u^{\lambda}) = \alpha = \frac{1}{2}$. This ends the proof of equality of the two rate regions. 
\end{proof}

In the sequel, we evaluate the effect of eavesdropping on such a BEC($e$)/BSC($p_2$) BC with a BSC($p$) eavesdropper.  

First note $\mathcal{C}_{std}$ the standard capacity region of the BC without an eavesdropper, $\mathcal{C}$ being its secrecy capacity region. We have that \cite{NairBECBSC2010}: 
\begin{equation} 
\mathcal{C}_{std}  \,: \left\{ \begin{array}{rcl}
R_1 &\leq& (1-e) \, h_2(x)   \ , \\
R_2 &\leq& 1 - h_2(p_2\star x) \ ,
\end{array}\right.  
\end{equation}
for some $x \in [0:0.5]$.

The presence of eavesdropper engenders an impediment on the sum rate given by $1- h_2(p)$, that does not depend on the choice of the channel parameters $(e, p_2)$. As such, it turns out that the channel to user $2$ ,i.e. BSC($p_2$) is very sensitive to such the BSC($p$) eavesdropper in that it could have zero admissible rate $R_2$ if the eavesdropper were to have a channel as good as to allow for $p = p_2$. However, and that's peculiar to the BEC($e$) channel, user $1$ always has strictly positive rates whatever the value of $p$, since $ e \leq 2p \leq h_2(p)$ and thus, a rate of $h_2(p) - 2p > 0 $ is always achievable. 

To illustrate this, we consider the following transmission scheme where $e = 2p$, i.e. the worst eavesdropper is considered for user $1$, and where we vary $p$ in the interval $[p_2:0.5]$. Fig \ref{FiGBECBSC} plots the obtained curves. As expected, the eavesdropper has no impediment on the available rates for both users when $p$ is close to $0.5$, however, as $p$ decreases, the gap between the standard capacity region and the secrecy capacity region increases, and the rate available at user $2$ decreases to zero whilst that of user $1$, stays above a given threshold. 

\begin{figure}[t]
   \centering
   \includegraphics[scale=0.4]{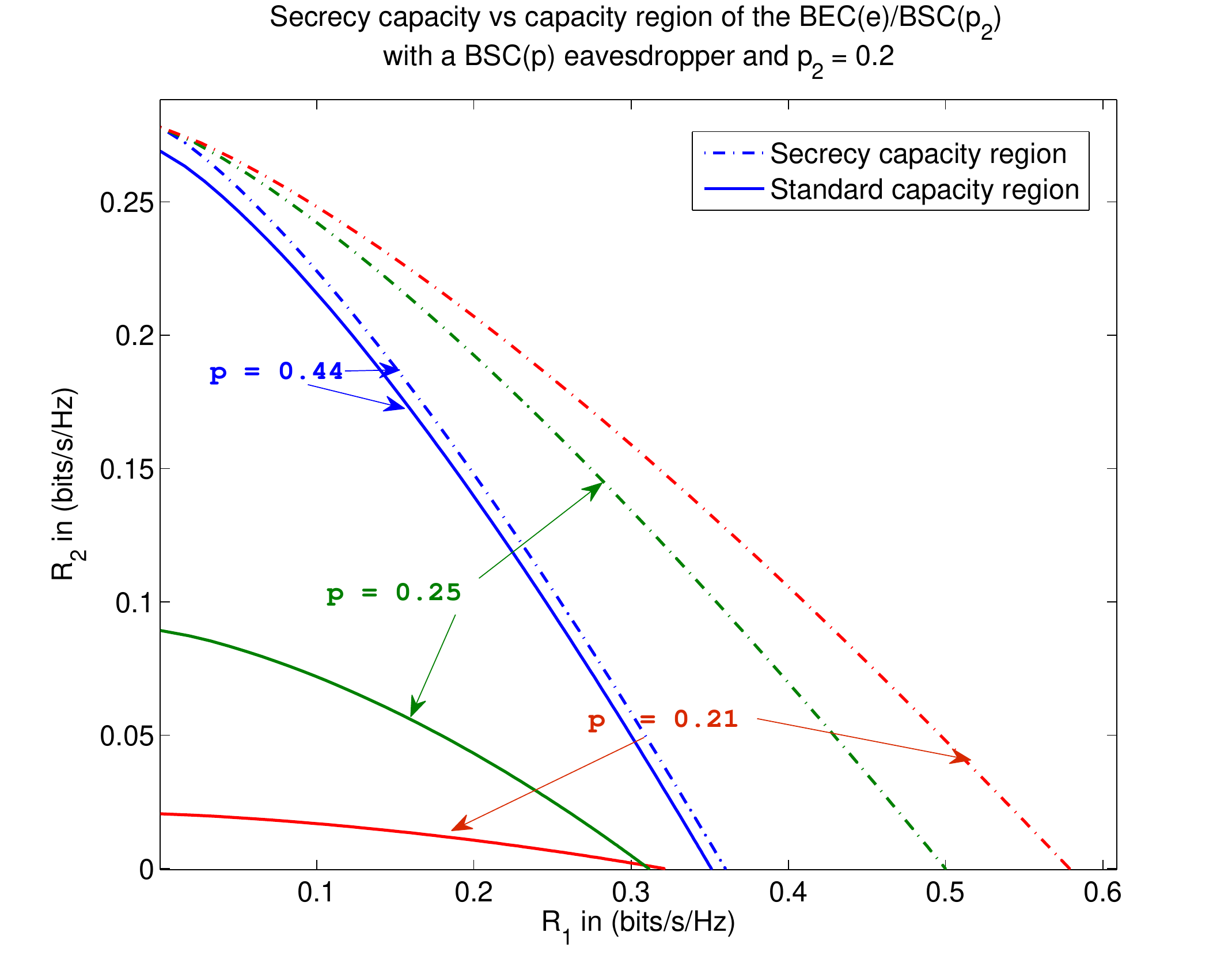} 
   \caption{Secrecy capacity region of the BC with BEC($e$)/BSC($p_2$) components and a BSC($p$) eavesdropper.}
   \label{FiGBECBSC} 
   \end{figure} 
   
%================= SECTION =================
   
\section{Proof of Theorem~\ref{TheoremOuter}: Outer Bound} \label{ProofOuterBound}
In this section, we prove the outer bound in Theorem~\ref{TheoremOuter},  since this rate region is symmetric in the rates $R_j$, $j \in \{ 1,2\}$, the constraints will be shown only for the following two single rates and two sum-rates: 
\begin{IEEEeqnarray}{rCl}
    R_1 &\leq& 	I( U_1; Y_1 |T V_1 ) - I(U_1; Z  |TV_1)  \ ,\\
    R_1 &\leq& 	I( U_1; Y_1 Y_2 |T V_1  V_2) - I(U_1; Z  |TV_1 V_2) \ , \\ 
	R_1 + R_2 &\leq& I(X; Y_2|T Z V_1) + I( U_1 S_1; Y_1 |T V_1 ) \nonumber \\
&& \qquad \qquad  - I(U_1 S_1; Z Y_2 |T V_1) \ ,\\ 
	R_1 + R_2 &\leq& I(X; Y_2|T Z V_1 V_2) + I( U_1 S_1 ; Y_1 Y_2 |T V_1 V_2 ) \nonumber \\
&& \qquad \qquad  - I(U_1 S_1 ; Z Y_2 |T V_1 V_2) \ . 
 \end{IEEEeqnarray}
%% -------------------------------------------------------------------------------------------------------------------%%
\subsection{Single rates' constraints}
By Fano's inequality we have that:
\begin{equation}
 n R_1 \leq I(W_1 ;Y_1^n ) + n\, \epsilon_n \ . 
\end{equation} 
Moreover, from the secrecy constraint: $I(W_1 ; Z^n ) \leq   n \,\epsilon_n $. 
Thus, one can write that: 
	\begin{IEEEeqnarray}{rCl}
   && n(R_1 -2  \epsilon_n) \nonumber \\
 &\leq& I(W_1 ;Y_1^n ) - I(W_1 ; Z^n ) \\ 
    &=& \sum^n_{i=1} \left[ I(W_1 ;Y_{1,i} | Y^{i-1}_{1} ) - I(W_1 ; Z_i| Z^{n}_{i+1}  )\right]   \\ 
    &\overset{(a)}{=}& \sum^n_{i=1} \left[ I(W_1 Z^{n}_{i+1};Y_{1,i} | Y^{i-1}_{1}) - I(W_1  Y^{i-1}_{1}; Z_i| Z^{n}_{i+1}  )\right] \,\,\,  \,\,\, \\ 
      &\overset{(b)}{=}&  \sum^n_{i=1} \left[ I(W_1;Y_{1,i} | Y^{i-1}_{1} Z^{n}_{i+1}) - I(W_1 ; Z_i| Y^{i-1}_{1}Z^{n}_{i+1}  ) \right]   \,\,\, \,\,\, 
	\end{IEEEeqnarray}
	where $(a)$ and $(b)$ follow both from the Csisz\'ar \& K\"orner's sum-identity~\eqref{sec:csiszarkorner}: 
\begin{IEEEeqnarray}{rCl} 
\sum^n_{i=1} \left[I(Z^{n}_{i+1};Y_{1,i} |  W_1 Y^{i-1}_{1} ) - I ( Y^{i-1}_{1} ; Z_i| W_1 Z^{n}_{i+1}  ) \right]  
&=& 0  \ ,  \qquad  \\
\sum^n_{i=1} \left[I(Z^{n}_{i+1};Y_{1,i} |  Y^{i-1}_{1} ) - I ( Y^{i-1}_{1} ; Z_i|   Z^{n}_{i+1}  ) \right]  &=& 0  \ . 
\end{IEEEeqnarray}
We then define: $U_{1,i} = W_1$, $V_{1,i} = Y^{i-1}_{1}$ and $T_i = Z^{n}_{i+1}$, which yields the first single rate constraint. 

In the same fashion, we can write the other single rates by treating the two outputs $Y_1$ and $Y_2$ together, i.e  $Y_1 \sim (Y_1, Y_2)$ letting $V_{2,i} = Y^{i-1}_{2}$. 
We end up with the couple of constraints: 
\begin{equation} 
	 \left\{\begin{array}{rcl}
    R_1 &\leq& 	I( U_1; Y_1 |T V_1 ) - I(U_1; Z  |TV_1)  \ ,\\
    R_1 &\leq& 	I( U_1; Y_1 Y_2 |T V_1  V_2) - I(U_1; Z  |TV_1 V_2) \ . 
    \end{array} \right.
\end{equation}
Furthermore, similar   all manipulations can be performed by starting from the Fano's inequality and secrecy requirement: 
\begin{equation}
 n R_1 \leq I(W_1 ;Y_1^n |W_2 ) - I(W_1; Z^n|W_2) + n\, \epsilon_n \ . 
\end{equation} 
Thus, we could condition over $ U_{2,i}=W_2 $ the two previous rate constraints to obtain: 
\begin{equation} 
	 \left\{\begin{array}{rcl} 
    R_1 &\leq& I( U_1; Y_1 |T V_1 U_2 ) - I(U_1; Z  |TV_1 U_2)  \ ,\\
    R_1 &\leq& I( U_1; Y_1 Y_2 |T V_1  U_2 V_2) - I(U_1; Z  |TV_1 U_2 V_2) \ .  
    \end{array} \right.
\end{equation} \vspace{-8mm}
\subsection{Sum-rate constraints} \label{SumRatesProof}
Let us start by Fano's inequality writing:
\begin{equation} 
    n \, R_1  \leq  I(W_1 ; Y_1^n ) - I(W_1 ; Y^n_2 Z^n )+  I(W_1 ; Y^n_2 Z^n )+  n\, \epsilon_n    \ . 
	\end{equation}
Then, combining with the following constraint obtained from Fano's inequality: 
\begin{equation}
  n\, R_2 \leq I(W_2 ;Y^n_2 Z^n |W_1) + n\, \epsilon_n \ , 
\end{equation}
we can write: 
\begin{IEEEeqnarray}{rCl}
 n \, ( R_1 + R_2 ) &\leq& I(W_1 ;Y_1^n ) - I(W_1 ;Y^n_2 Z^n ) \nonumber \\
&& \qquad  +  I(W_1 W_2 ;Y^n_2 Z^n )+  2 n\, \epsilon_n    \ .   
\end{IEEEeqnarray}
   
Now, let us elaborate on that: 
\begin{IEEEeqnarray}{rCl}
&& I(W_1 ;Y_1^n ) - I(W_1 ;Y^n_2 Z^n ) \IEEEnonumber  \\
 &=& \sum^n_{i=1} \left[ I(W_1;Y_{1,i} | Y^{i-1}_{1}  ) - I(W_1 ; Y_{2,i} Z_i|  Y^{n}_{2,i+1} Z^{n}_{i+1} )\right]  \qquad  \\
    & \overset{(a)}{=} &  \sum^n_{i=1} \bigl[ I(W_1 Y^{n}_{2,i+1} Z^{n}_{i+1};Y_{1,i}| Y^{i-1}_{1} ) \nonumber \\
&& \qquad \qquad - I(W_1 Y^{i-1}_{1}  ; Y_{2,i} Z_i |  Y^{n}_{2,i+1} Z^{n}_{i+1}  ) \bigr] \\ 
     &=&  \sum^n_{i=1} \bigl[ I(W_1 Y^{i-1}_{1} Y^{n}_{2,i+1} Z^{n}_{i+1};Y_{1,i} )\nonumber \\
&& \qquad \qquad  - I(W_1 Y^{i-1}_{1} Y^{n}_{2,i+1} Z^{n}_{i+1}; Y_{2,i} Z_i  )   \IEEEnonumber \\ 
     && \vphantom{\sum^n_{i=1}}\qquad \qquad  +  I( Y^{n}_{2,i+1} Z^{n}_{i+1}; Y_{2,i} Z_i  )  - I( Y^{i-1}_{1} ;Y_{1,i} )  \bigr]  \ , 
	\end{IEEEeqnarray}
	where $(a)$ is again a consequence of Csisz\'ar \& K\"orner's sum-identity~\eqref{sec:csiszarkorner}: 
\begin{IEEEeqnarray}{rCl}	
&& 	\sum^n_{i=1}    I(Z^{n}_{i+1};Y_{1,i} |  W_1 Y^{i-1}_{1} ) \nonumber  \\
&=& \sum^n_{i=1} I ( Y^{i-1}_{1} ; Y_{2,i}Z_i| W_1 Y^{n}_{2,i+1} Z^{n}_{i+1}  )     \ .  
\end{IEEEeqnarray}	
    As for the other term, note that: 
\begin{IEEEeqnarray}{rCl}
     I(W_1 W_2;Y^n_2 Z^n )  &= &  \sum^n_{i=1} \bigl[  I(W_1 W_2Y^{n}_{2,i+1} Z^{n}_{i+1} ;Y_{2,i} Z_i   )\IEEEnonumber\\
    && \qquad \quad   -  I( Y^{n}_{2,i+1} Z^{n}_{i+1}; Y_{2,i} Z_i  ) \bigr]  \ . \,\,\,
\end{IEEEeqnarray}  
	Looking at the first term of the last equality: 
\begin{IEEEeqnarray}{rCl} 
      && \sum^n_{i=1}  I(W_1 W_2 Y^{n}_{2,i+1} Z^{n}_{i+1} ;Y_{2,i} Z_i ) \IEEEnonumber\\
    &=& \sum^n_{i=1}  \left[ I(W_1 W_2 Y^{n}_{2,i+1} Z^{n}_{i+1} Z^{i-1} ; Y_{2,i} Z_i ) \right. \IEEEnonumber\\
&& \qquad  \left. -  I(Z^{i-1} ; Y_{2,i} Z_i  |W_1 W_2 Y^{n}_{2,i+1} Z^{n}_{i+1}  ) \right]\\
    &\overset{(a)}{=}& \sum^n_{i=1} \left[  I(W_1 W_2 Y^{n}_{2,i+1} Z^{n}_{i+1} Z^{i-1} ; Y_{2,i} Z_i )\right. \IEEEnonumber\\
&& \qquad  \left. -  I(Y^{n}_{2,i+1} Z^{n}_{i+1}  ;  Z_i  |W_1 W_2 Z^{i-1} ) \right]\\ 
    &=& \sum^n_{i=1}  \left[ I(W_1 W_2  Z^{i-1} ;  Z_i )\right. \IEEEnonumber\\
&& \qquad  \left. +   I(W_1 W_2 Y^{n}_{2,i+1} Z^{n}_{i+1} Z^{i-1} ; Y_{2,i}|  Z_i)\right] \\ 
    &=& \sum^n_{i=1}  \left[ I(W_1 W_2   ;  Z_i |  Z^{i-1}) + I( Z^{i-1}; Z_i  ) \right. \IEEEnonumber\\
&& \qquad  \left.+ I(W_1 W_2 Y^{n}_{2,i+1} Z^{n}_{i+1} Z^{i-1} ; Y_{2,i}|  Z_i) \right] \ .
     \end{IEEEeqnarray}	
    Here, $(a)$ is a consequence of Csisz\'ar \& K\"orner's sum-identity~\eqref{sec:csiszarkorner} but between the outputs $Z$ and $(Y_2, Z)$: 
\begin{IEEEeqnarray}{rCl}	
&& 	\sum^n_{i=1}   I(Z^{i-1} ; Y_{2,i} Z_i  |W_1 W_2 Y^{n}_{2,i+1} Z^{n}_{i+1}  )  \nonumber \\
&=& \sum^n_{i=1}   I(Y^{n}_{2,i+1} Z^{n}_{i+1}  ;  Z_i  |W_1 W_2 Z^{i-1} )    \ .
\end{IEEEeqnarray}	
Using the secrecy constraint, one can then notice that: 
\begin{equation}
     \sum^n_{i=1}  I(W_1 W_2 ; Z_i |  Z^{i-1}) =   I(W_1 W_2;  Z^n ) \leq n  \, \epsilon_n  \ . 
\end{equation}
Moreover, observe  that: 
\begin{equation}
     \sum^n_{i=1} I( Z^{i-1}; Z_i  ) =  \sum^n_{i=1}  I( Z^n_{i+1};  Z_i  )  \ , 
\end{equation}
and 
\begin{IEEEeqnarray}{rCl}
  &&   I(W_1 W_2 Y^{n}_{2,i+1} Z^{n}_{i+1} Z^{i-1} ; Y_{2,i}|  Z_i)   \nonumber \\ 
  & \leq&	  I(W_1 W_2 Y^{n}_{2,i+1} Z^{n}_{i+1}  Y_1^{i-1} Z^{i-1} ; Y_{2,i}|  Z_i) \ .
\end{IEEEeqnarray}
The sum-rate can be then bounded as follows: 
\begin{IEEEeqnarray}{rCl}
   && n ( R_1 + R_2  -  2 \epsilon_n)  \IEEEnonumber\\  
   &\leq&  \vphantom{\sum^n} I(W_1 ;Y_1^n ) - I(W_1 ;Y^n_2 Z^n )   +  I(W_1 W_2 ;Y^n_2 Z^n )\,\,\, \\  
      &\leq& \sum^n_{i=1}\left[  I(W_1 Y^{i-1}_{1} Y^{n}_{2,i+1} Z^{n}_{i+1};Y_{1,i} ) \right. \IEEEnonumber\\
&&  \qquad \vphantom{\sum_{i=1}}  \left. - I(W_1 Y^{i-1}_{1} Y^{n}_{2,i+1} Z^{n}_{i+1}; Y_{2,i} Z_i  ) \right. \IEEEnonumber\\ 
     &&  \vphantom{\sum_{i=1}} \qquad + I(W_1 W_2 Y^{n}_{2,i+1} Z^{n}_{i+1} Y_1^{i-1} Z^{i-1} ; Y_{2,i}|  Z_i)  \IEEEnonumber\\ 
      &&  \vphantom{\sum_{i=1}} \qquad + \left.     I( Z^n_{i+1};  Z_i  ) - I( Y^{i-1}_{1} ;Y_{1,i} ) \right]   \  .
	\end{IEEEeqnarray}	
And to end, we use the following remarks: 
\begin{IEEEeqnarray}{rCl}
     && \sum^n_{i=1} \left[I( Z^n_{i+1}; Z_i ) - I( Y^{i-1}_{1} ;Y_{1,i} ) \right]   \IEEEnonumber\\
 &=&    \sum^n_{i=1} \left[ I( Y^{i-1}_{1} Z^n_{i+1}; Z_i ) - I( Y^{i-1}_{1}  Z^n_{i+1};Y_{1,i} )\right] \\
      &=& \sum^n_{i=1}\left[ I( Y^{i-1}_{1} Z^n_{i+1}; Y_{2,i} Z_i ) - I( Y^{i-1}_{1}  Z^n_{i+1};Y_{1,i} ) \right. \IEEEnonumber \\
      &&\quad \quad \left.- I( Y^{i-1}_{1} Z^n_{i+1}; Y_{2,i} | Z_i )\right] \ . 
\end{IEEEeqnarray}	
Thus, combining with the previous equality, we end up with: 
\begin{IEEEeqnarray}{rCl}
   && n  ( R_1 + R_2 -  2   \epsilon_n)  \IEEEnonumber \\
  &\leq& \vphantom{\sum^n} I(W_1 ;Y_1^n ) - I(W_1 ;Y^n_2 Z^n )+  I(W_1 W_2 ;Y^n_2 Z^n ) \\  
      &\leq& \sum^n_{i=1} \left[ I(W_1 Y^{i-1}_{1} Y^{n}_{2,i+1} Z^{n}_{i+1};Y_{1,i} ) - I( Y^{i-1}_{1}  Z^n_{i+1};Y_{1,i} )\right. \IEEEnonumber \\
      &&\quad  \vphantom{\sum_{i=1}} \left. - I(W_1 Y^{i-1}_{1} Y^{n}_{2,i+1} Z^{n}_{i+1}; Y_{2,i} Z_i  )+   I( Y^{i-1}_{1} Z^n_{i+1}; Y_{2,i} Z_i )  \right. \IEEEnonumber \\
      &&\quad \vphantom{\sum_{i=1}}  + \left.  I(W_1 W_2 Y^{n}_{2,i+1} Y_1^{i-1}  Z^{n}_{i+1}  Z^{i-1} ; Y_{2,i}| Z_i )  \right. \IEEEnonumber \\
      &&\quad \vphantom{\sum_{i=1}} \left. - I( Y^{i-1}_{1} Z^n_{i+1}; Y_{2,i} | Z_i ) \right]   \\
     &\overset{(a)}{\leq}& \sum^n_{i=1} \left[ I(W_1 Y^{n}_{2,i+1} ;Y_{1,i} | Y^{i-1}_{1}Z^{n}_{i+1})\right. \IEEEnonumber \\
      &&\quad \left. \vphantom{\sum_{i=1}} - I(W_1 Y^{n}_{2,i+1} ; Y_{2,i} Z_i | Y^{i-1}_{1} Z^{n}_{i+1} )\right.  \IEEEnonumber \\  
  &&\quad \vphantom{\sum_{i=1}}  + \left. I(X_i; Y_{2,i} |  Z_i Y_1^{i-1}  Z^{n}_{i+1} )\right] + 2  n\, \epsilon_n  \ ,    
	\end{IEEEeqnarray}
	where $(a)$ is a consequence of introducing the input $X_i$:
	\begin{IEEEeqnarray}{rCl}
	&& I(W_1 W_2 Y^{n}_{2,i+1}  Z^{i-1} ; Y_{2,i}| Z_i Y_1^{i-1}  Z^{n}_{i+1} )   \nonumber \\ 
  &\leq& I(X_i; Y_{2,i} |  Z_i Y_1^{i-1}  Z^{n}_{i+1} )\ . 
	\end{IEEEeqnarray}
	Letting:  $S_{1,i} = Y^{n}_{2,i+1}$, $U_{1,i} = W_1$, $V_{1,i} = Y^{i-1}_{1}$, and $T_i = Z^{n}_{i+1}$, and noting that: by resorting to a standard time-sharing argument we end up with the following single-letter constraint: 
\begin{IEEEeqnarray}{rCl}
     R_1 + R_2  &\leq&   I(U_1 S_1 ;Y_{1 }| V_1 T) -   I(U_1 S_1 ;Y_{2 }Z | V_1 T )    \nonumber \\ 
  && \qquad    +  I(X_i; Y_{2 } |Z V_1 T)  \ .    
	\end{IEEEeqnarray}

Similarly, we can show the same sum-rate constraint, by replacing the output $Y_1$ with  the two outputs $(Y_1 Y_2)$, which results in: 
\begin{IEEEeqnarray}{rCl} 
	R_1 + R_2 &\leq& I(X; Y_2|T Z V_1 V_2) + I( U_1 S_1 ; Y_1 Y_2 |T V_1 V_2 )   \nonumber \\ 
  && \qquad  - I(U_1 S_1 ; Z Y_2 |T V_1 V_2) \ .  
\end{IEEEeqnarray}

\subsection{Proof of Corollary~\ref{CorollaryOuter}}\label{EliminationS}
In the previous section, we found that an outer bound on the secrecy region for the Wiretap BC can be obtained by considering only the constraints: 
\begin{IEEEeqnarray}{rCl}
    R_1 &\leq& 	I( U_1; Y_1 | T  V_1) - I(U_1; Z  |TV_1 ) \ , \\
    R_2 &\leq&  I( U_2; Y_2 |T V_2 ) - I(U_2; Z  |TV_2)  \ ,\\
	R_1 + R_2 &\leq& I(X; Y_2|T  Z V_1 ) + I( U_1 S_1 ; Y_1 |T V_1  ) \IEEEnonumber \\
      &&\quad  - I(U_1S_1; Z Y_2 |T V_1 ) \ , 	\\
	R_1 + R_2 &\leq& I(X; Y_1|T Z V_2) + I( U_2 S_2; Y_2 |T V_2 )  \IEEEnonumber \\
      &&\quad - I(U_2 S_2; Z Y_1 |T V_2)   \ .
	\end{IEEEeqnarray} 
An important claim is then that the auxiliary rvs $S_1$ and $S_2$ can be eliminated with no impediment to the rate region. Since the region is symmetric in $R_1$ and $R_2$, we only show the claim for $S_1$. We are looking for a random variable $U_1^\star$ such that we can write: 
\begin{IEEEeqnarray}{rCl}
    R_1 &\leq& 	I( U_1^\star; Y_1 | TV_1  ) - I(U_1^\star; Z | TV_1 )  \ , \\
	R_1 + R_2 &\leq& I(X; Y_2| T Z V_1 ) + I( U_1^\star  ; Y_1 |T V_1  )  \IEEEnonumber \\
      &&\quad  - I(U_1^\star; Z Y_2 |T V_1 )	  \ .
	\end{IEEEeqnarray}
To see this, define the two following functions: 
\begin{IEEEeqnarray*}{rCl}
     f_1(Q) & \triangleq& I( U_1; Y_1 | TV_1 ) - I(U_1; Z  |TV_1)  \IEEEnonumber \\
      &&\qquad - I( Q ; Y_1 |TV_1 ) + I(  Q ; Z  |TV_1)  \ , \\
     f_2(Q) & \triangleq& I( U_1 S_1 ; Y_1 |T  V_1) - I(U_1 S_1 ; Y_2 Z |TV_1)  \IEEEnonumber \\
      &&\qquad  - I( Q; Y_1 | TV_1) + I(Q ;  Y_2 Z |TV_1)  \ .
	\end{IEEEeqnarray*}
We note first that: 
\begin{equation}
 	f_1(U_1) = 0 \quad , \quad f_2(U_1S_1) = 0 \ .
\end{equation}
Moreover, 
\begin{IEEEeqnarray}{rCl}
 f_1(U_1 S_1) &+& f_2(U_1)  \IEEEnonumber \\
      &=& - I(U_1 S_1 ; Y_2 | TZV_1   ) + I(U_1; Y_2 | TZV_1 )\,\,\,\,\,\,\\
     &=& - I(S_1 ; Y_2 |T Z U_1 V_1  ) \\
     &\leq& 0 \ . 
	\end{IEEEeqnarray}
Therefore, either  $ f_1(U_1 S_1) \leq 0$ and thus, letting $U_1^\star= (U_1 S_1)$ will not reduce the region, or $ f_2(U_1) \leq 0$ and in this case $U_1^\star= U$ allows us to prove our claim. 
The same holds for the other couple of constraints on $R_2$ and $R_1 + R_2$. 

%================= SECTION =================

\section{Proof of Theorem \ref{TheoremInner}: Inner Bound} \label{ProofInnerBound}
	In this section, we prove the achievability of the inner bound stated in Theorem \ref{TheoremInner}. Let $R_1$ and $R_2$ denote the information rates. Let $T$ be any the time sharing random variable. The coding argument is as follows. 
	
	\subsection{Code generation, encoding and decoding procedures}
	
	%% -------------------------------------------------------------------------------------------------------------------%%
	\subsubsection{Rate splitting}
	We split the message intended to each user of rate $R_j$ into two sub-messages: one of  rate $\bar{R}_j= R_j - R_{0j}$ that will be decoded only by the user, and one of rate $R_{0j}$ that will be carried through the common message. Thus in stead of transmitting the message pair $(  w_1, w_2)$, we  transmit the triple $(\bar{w}_0, \bar{w}_1, \bar{w}_2)$.
	\begin{equation} \label{ConstraintRS} 
	\left\{ \begin{array}{rcl}
    \bar{R}_0 &\triangleq & R_{01} + R_{02}   \ ,  \\ 
    \bar{R}_j &\triangleq & R_j - R_{0j} \geq 0  \ .
	\end{array}\right.
	\end{equation}
	%% -------------------------------------------------------------------------------------------------------------------%%
	\subsubsection{Codebook generation} 	
	Generate $2^{n T_0}$ sequences $q^n(s_0)$ following 
	\begin{equation} 
	P_Q^n(q^n(s_0)) = \prod_{i=1}^n P_{Q}(q^n_i(s_0)) \ , 
	\end{equation}
	where $T_0 \geq \bar{R}_0$ and map these in $2^{n\bar{R}_0}$ bins indexed by $\bar{w}_0$: $\mathcal{B}^n_0(\bar{w}_0)$.  
	
	For each $s_0 \in [1:2^{nT_0}]$ and for each $j\in\{1,2\}$, generate $2^{nT_j}$ sequences $u_j^n( s_0, s_j)$ following      	
	\begin{equation}  
	P_{U_j|Q} ^n(u_j^n( s_0, s_j) | q^n(s_0)) = \prod_{i=1}^n P_{U_j|Q}( u_{j,i}^n( s_0, s_j) | q^n_i(s_0)) \ . 
	\end{equation} 
	Map these sequences in $2^{n\bar{R}_j}$ bins indexed by $\bar{w}_j$: $\mathcal{B}^n_j(s_0,\bar{w}_j) $ and consisting in $2^{n(T_j-\bar{R}_j)}$ n-sequences. Each of these bins are divided into $2^{n\tilde{R}_j}$ sub-bins indexed by $l_j$: $\mathcal{B}^n_j(s_0,\bar{w}_j,l_j)$, thus each bin contains $2^{n(T_j-\bar{R}_j-\tilde{R}_j)}$ sequences where $0 \leq \tilde{R}_j \leq T_j-\bar{R}_j$. 
	
	The codebook consisting of all the bins is known to all terminals, including the eavesdropper.

	\begin{figure}[t]
   \centering
   \includegraphics[scale=.55]{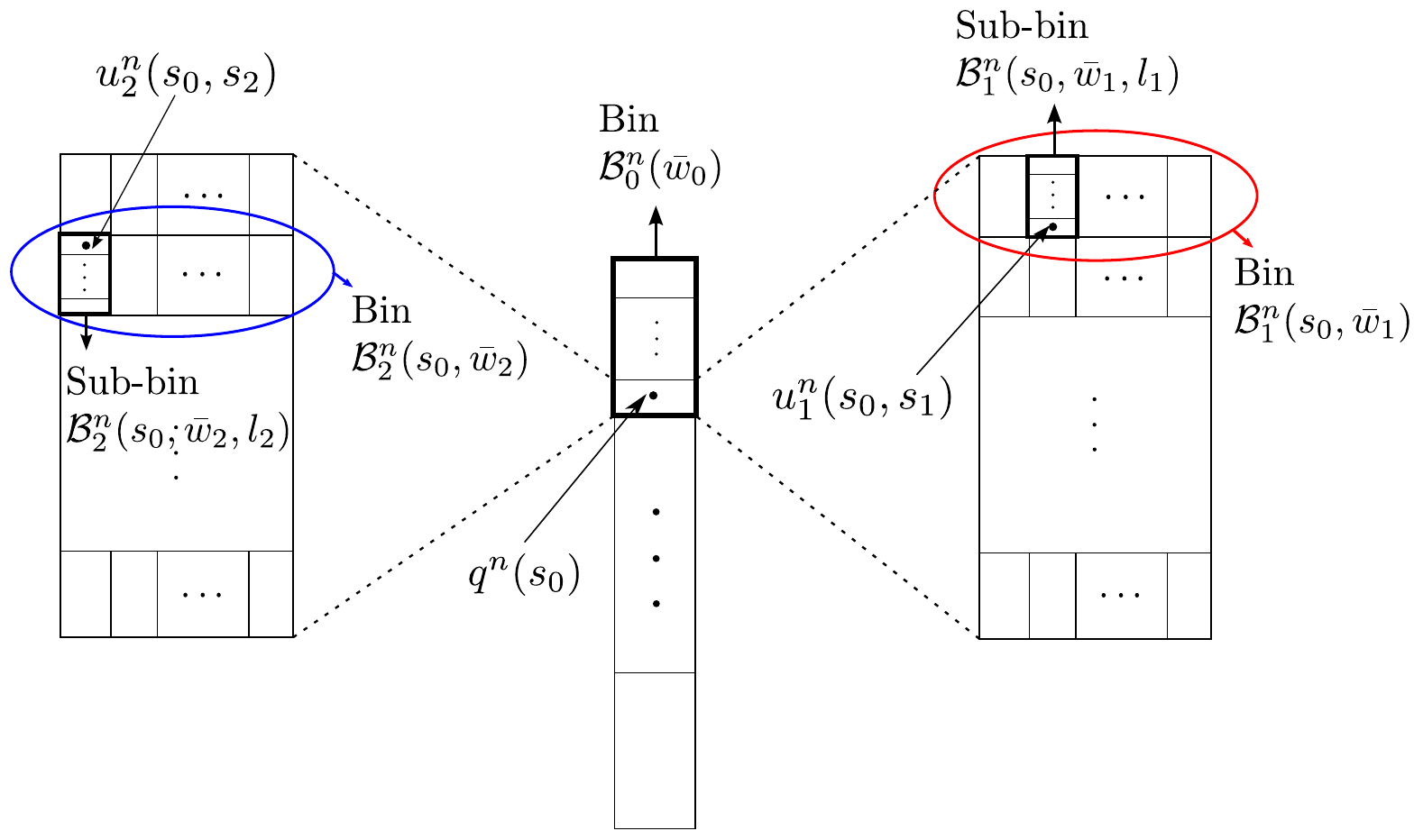} 
   \caption{Codebook generation and encoding.}
   \label{WBCEncoding} 
   \end{figure}
   %% -------------------------------------------------------------------------------------------------------------------%%%
	\subsubsection{Encoding} 
	Fig.~\ref{WBCEncoding} plots the encoding operation.  To send $(\bar{W}_0, \bar{W}_1, \bar{W}_2)$, the encoder selects at random an index $s_0$ such that $q^n(s_0) \in \mathcal{B}^n_0(\bar{w}_0)$. Then, in the product bin $\mathcal{B}^n_1(s_0,\bar{w}_1) \times \mathcal{B}^n_2(s_0,\bar{w}_2)$, it chooses at random a pair of sub-bins $\mathcal{B}^n_j(s_0,\bar{w}_1,l_1)$ and $\mathcal{B}^n_2(s_0,\bar{w}_2,l_2)$ indexed by $l_1$ and $l_2$. In the corresponding product sub-bin, it looks for a pair of sequences indexed with $s_1$ and $s_2$ satisfying: 
	\begin{equation}
		\left(  q^n(s_0) , u_1^n( s_0 , s_1), u_2^n( s_0 , s_2) \right) \in \typ{Q U_1 U_2} \ . 
	\end{equation}
	Based on the Mutual Covering Lemma~\cite{cover2006elements}, the encoding will succeed if the following inequalities hold: 
	\begin{equation}
	\left\{\begin{array}{rcl} 
	 T_1 - (\bar{R}_1 +  \tilde{R}_1) + T_2 - (\bar{R}_2 + \tilde{R}_2)  &> & I(U_1; U_2 |Q)  \ , \\ 
	 0  &\leq& \tilde{R}_1 \leq T_1-\bar{R}_1 \ , \\ 
	 0 &\leq& \tilde{R}_2 \leq T_2-\bar{R}_2 \ . 
	 \end{array}\right.
	\end{equation}
	
	%% -------------------------------------------------------------------------------------------------------------------%%%
	\subsubsection{Decoding}
	Upon receiving $y_j^n$, decoder $j$ looks  jointly for a pair of indices $(s_0,s_j)$ such that: 
	\begin{equation}
		\left(  q^n(s_0) , u_j^n( s_0 , s_j), y^n_j\right) \in
		\typ{Q U_j Y_j} \ . 
	\end{equation}
	From the decoded indices $s_0$ and $s_j$, it can infer the initial values of both $\bar{W}_0$ and $\bar{W}_j$. 
	
	Based on Lemma~\ref{lem:jointTypicality}, the error probability can be made arbitrarily small provided that: 
	\begin{equation}
		 \left\{\begin{array}{rcl} 
		 T_j &\leq& I(U_j; Y_j | Q) \ , \\
		 T_j + T_0 &\leq& I(Q U_j; Y_j) \ .
		\end{array}\right.
		\end{equation}
	%% -------------------------------------------------------------------------------------------------------------------%%%
	
	\subsection{Equivocation analysis}
	We find conditions on the rates $T_0 , T_1 , T_2 $ and $\tilde{R}_1 , \tilde{R}_2$ to achieve perfect secrecy for all message triples $(\bar{W}_0, \bar{W}_1, \bar{W}_2)$. 
	
	To this end, we first note that it suffices to find conditions for which $\frac{1}{n}I(\bar{W}_0\bar{W}_1 \bar{W}_2 ; Z^n |\mathcal{C})$ can be made arbitrarily small where $ \mathcal{C} $ denotes the codebook used in the transmission, the latter constraint leading to the individual secrecy requirements being fulfilled. 
	
	Note that: 
    \begin{IEEEeqnarray}{rCl}
	  &&  I(\bar{W}_0 \bar{W}_1 \bar{W}_2 ; Z^n | \mathcal{C} ) \IEEEnonumber \\
      &=& n(\bar{R}_0 + \bar{R}_1 + \bar{R}_2) - H(\bar{W}_0 \bar{W}_1\bar{W}_2 | Z^n , \mathcal{C} )  \\
	   &\overset{(a)}{=}& n(\bar{R}_0 + \bar{R}_1 + \bar{R}_2) - H(S_0 S_1 S_2 | Z^n ,\mathcal{C} )\nonumber \\
	   && \qquad \qquad + H(S_0 S_1 S_2 | Z^n \bar{W}_0 \bar{W}_1 \bar{W}_2 ,\mathcal{C}) \label{WTP1}   \ , 
	\end{IEEEeqnarray}	
	where (a) follows from that, knowing the codebook, the sent messages are deterministic functions of the binning indices chosen. 
	 
    We first start by giving a lower bound to $H(S_0S_1S_2 | Z^n ,\mathcal{C})$.    
    Let us write: 
    \begin{IEEEeqnarray}{rCl}
		&& H(S_0S_1S_2 | Z^n, \mathcal{C} ) \IEEEnonumber\\
	   &=& H(S_0 | Z^n ,\mathcal{C}  ) + H(S_1S_2 | Z^n ,S_0 ,\mathcal{C} ) \\
	   &=& H(S_0 |\mathcal{C} ) - I(S_0 ; Z^n |\mathcal{C}  ) + H(S_1S_2 | S_0 ,\mathcal{C} ) \nonumber \\ 
	   	&& \quad -  I(S_1 S_2  ; Z^n |S_0 ,\mathcal{C} ) \\ 
	   &=& n T_0 - I(S_0 ; Z^n|\mathcal{C} ) +  H(S_1S_2 | S_0 ,\mathcal{C} ) \nonumber \\ 
	   	&& \quad- I(S_1 S_2  ; Z^n |S_0 ,\mathcal{C} ) \\
	   &=& n(T_0 + T_1 + T_2) - I(S_0 ; Z^n|\mathcal{C} ) \nonumber \\ 
	   	&& \quad- I(S_1;S_2| S_0 ,\mathcal{C} ) - I(S_1 S_2  ; Z^n |S_0 ,\mathcal{C} ) \\
	  % &=& n(T_0 + T_1 + T_2) - I(Q^n S_0 ; Z^n|\mathcal{C} ) - I(U_1^n S_1; U_2^n S_2| Q^n S_0 ,\mathcal{C} ) - I(U_1^n U_2^n S_1 S_2  ; Z^n |Q^n,S_0 ,\mathcal{C} ) \\
	  %&\geq& n(T_0 + T_1 + T_2) - I(Q^n S_0 ; Z^n | \mathcal{C} ) - I(U_1^n S_1; U_2^n S_2 S_0 | Q^n ,\mathcal{C}) - I(U_1^n U_2^n S_1 S_2 S_0 ; Z^n |Q^n,\mathcal{C} ) \\
	   &\overset{(a)}{=}&  n(T_0 + T_1 + T_2) - I(Q^n ; Z^n | \mathcal{C}) \nonumber \\ 
	   	&& \quad - I(U_1^n ; U_2^n | Q^n , \mathcal{C}) - I(U_1^n U_2^n ; Z^n |Q^n, \mathcal{C}) \ , 
	\end{IEEEeqnarray}	
   where (a) follows similarly from the fact that, knowing the codebook, the sent sequences are functions of the chosen binning indices. 
   
	The next lemma provides the main result for carrying on with the analysis. 
	\begin{lemma}\label{Inequalities}
	Assuming the codebook generation presented before, the next inequalities hold true: 
    \begin{IEEEeqnarray}{rCl}
		 I(Q^n ; Z^n | \mathcal{C}) &\leq& n I(Q;Z) + n \, \epsilon_n  \ , \\
		 I(U_1^n ; U_2^n | Q^n , \mathcal{C}) &\leq& n I(U_1; U_2 |Q) + n \, \epsilon_n \ , \\
		  I(U_1^n  U_2^n; Z^n |Q^n , \mathcal{C}) &\leq& n I(U_1 U_2; Z |Q) + n \,  \epsilon_n  \ .   
		\end{IEEEeqnarray}
	\end{lemma}
	\begin{IEEEproof} The proof of this lemma is presented in Appendix \ref{ProofLemmaInequalities}. \end{IEEEproof}
Lemma~\ref{Inequalities}  allows us thus to write: 
	\begin{IEEEeqnarray}{rCl}\label{WTP2}
	 \dfrac{1}{n}H(S_0  S_1 S_2 | Z^n, \mathcal{C} ) &\geq & T_0 + T_1 + T_2 - I(Q U_1 U_2; Z) \nonumber \\
	 &&    - I(U_1; U_2 | Q)  \ .
	 \end{IEEEeqnarray}
	Now, let us upper bound the remainder term to be studied: $H(S_0S_1S_2 | Z^n \bar{W}_0 \bar{W}_1 \bar{W}_2 ,\mathcal{C})$. 
	
	The following Lemma is useful to carry on with the proof. 
	\begin{lemma}\label{LemmaUpperBound}
	Assuming the same coding scheme presented before, then 
	\begin{IEEEeqnarray}{rCl}
	&& \limsup_{n \rightarrow \infty }\dfrac{1}{n}  H(S_0 S_1 S_2 | Z^n \bar{W}_0 \bar{W}_1 \bar{W}_2 ,\mathcal{C} ) \nonumber \\
	&& \qquad \leq  \max \left\{0 , I_1, I_2, I_3 , I_4\right\} \ , 
	\end{IEEEeqnarray}
	 where 
	 \begin{IEEEeqnarray}{rCl}
	I_1 &=& T_1 - \bar{R}_1 - I(U_1; Z U_2 | Q) \ , \\ 
	I_2 &=& T_2 - \bar{R}_2 - I(U_2; Z U_1 | Q)  \ , \\
	I_3 &=& T_1 - \bar{R}_1 + T_2 - \bar{R}_2  \IEEEnonumber \\ 
	  &&   \qquad  - I(U_1U_2  ; Z | Q) - I(U_1; U_2 |Q) \ ,  \\
	I_4 &=& T_0 - \bar{R}_0 + T_1 - \bar{R}_1 + T_2 - \bar{R}_2  \IEEEnonumber \\ 
	  &&   \qquad  - I(QU_1U_2  ; Z) - I(U_1; U_2 |Q) \ .
	\end{IEEEeqnarray}	 
	\end{lemma}
	\begin{IEEEproof} This Lemma is proved in Appendix \ref{ProofLemmaUpperBound}.\end{IEEEproof}
	
	As a conclusion of this lemma, and combining \eqref{WTP1} and \eqref{WTP2} we can conclude that: 
	 \begin{IEEEeqnarray}{rCl}
	&& \dfrac{1}{n} I(\bar{W}_0 \bar{W}_1 \bar{W}_2 ; Z^n | \mathcal{C} ) -  \epsilon_n   \IEEEnonumber \\
	 &\leq &  \bar{R}_0 + \bar{R}_1 + \bar{R}_2 - (T_0 + T_1 + T_2)   + I(QU_1U_2  ; Z)\IEEEnonumber \\ 
	 && \quad +  I(U_1; U_2 |Q) \max\left\{0 , I_1, I_2 , I_3 , I_4\right\}  \\ 
	  &=& \max \big\{ \bar{R}_0 + \bar{R}_1 + \bar{R}_2 - (T_0 + T_1 + T_2)  \IEEEnonumber \\
      &&\qquad \qquad  + I(QU_1U_2  ; Z) + I(U_1; U_2 |Q) \  , \IEEEnonumber \\ 
	  &&   \qquad  \bar{R}_0 - T_0 + \bar{R}_2 - T_2 + I(QU_2;Z) \  , \IEEEnonumber\\
	  &&  \qquad  \bar{R}_0 - T_0 + \bar{R}_1 - T_1 + I(QU_1;Z) \ ,  \IEEEnonumber \\ 
	  &&   \qquad \bar{R}_0 - T_0 + I(Q; Z) \ ,   \    0    \big\}  \ . 
	\end{IEEEeqnarray}	
	Hence,  full secrecy is guaranteed by forcing all operands in the max term to be  less  than zero.  
	
By collecting all inequalities and applying FME  on the rates $R_{01}$ and $R_{02}$ (see details in Appendix~\ref{appendix-missing}), we obtain the desired rate region: 
 \begin{IEEEeqnarray}{rCl}
	    R_1 &\leq& I( Q U_1; Y_1 ) - I( Q U_1; Z )  \ ,\\
		R_2 &\leq& I( Q U_2; Y_2 ) - I( Q U_2; Z )   \ ,\\
		R_1 + R_2 &\leq& I( U_1; Y_1| Q ) + I( Q U_2; Y_2 ) \IEEEnonumber \\ 
	  &&   \qquad - I(Q U_1 U_2 ; Z) - I(U_1; U_2|Q )  \ ,\\
		R_1 + R_2 &\leq& I( U_2; Y_2| Q ) + I( Q U_1; Y_1 ) \IEEEnonumber \\ 
	  &&   \qquad  - I(Q U_1 U_2 ; Z)  - I(U_1; U_2|Q ) \ ,\\
		R_1 + R_2 &\leq& I( Q U_1; Y_1 ) + I( Q U_2; Y_2 )- I(Q U_1 U_2 ; Z) \IEEEnonumber \\ 
	  &&   \qquad   -  I(U_1; U_2|Q )- I(Q;Z)   \ .  
		\end{IEEEeqnarray}
		Obviously, the time sharing variable $T$ can be added and thus, the achievability of the region~\eqref{InnerBoundWBC} is proved. \qed 
		
	%	As stated previously, the bound in \eqref{InnerBoundWBC} collapses to Marton's inner bound with a common message when the eavesdropper is absent, $Z = \emptyset$.  

%% -------------------------------------------------------------------------------------------------------------------%%%

\section{Summary and Discussion}
In this work, we investigated the secrecy capacity region of the general memoryless two-user Wiretap Broadcast Channel (WBC). We derived a novel outer bound which implies, to the best of our knowledge, all known capacity results in the corresponding setting while by removing secrecy constraints it performs as well as the best-known outer bound for the general Broadcast Channel (BC). An inner bound on the secrecy capacity region of the WBC was also derived by simply using existent encoding techniques based on random binning and stochastic encoders. These  bounds allowed us to characterize  the secrecy capacity region of several classes of channels, including  the deterministic BC with a general eavesdropper, the semi-deterministic BC with a more-noisy eavesdropper and the less-noisy BC with a degraded eavesdropper, as well as some classes of ordered BCs previously studied. Furthermore, the secrecy capacities of the BC with BEC/BSC components and a BSC eavesdropper, as well as the product of two inversely ordered BC with a degraded eavesdropper were also characterized. 

In the same spirit of Corollary~\ref{CorollaryOuter}, a more general study of the role of the auxiliary variables of the outer bound in Theorem~\ref{TheoremOuter}  may lead to the characterization of capacity for other classes of Wiretap BCs and this will be object of future work. 
 
\appendices
\section{Useful Notions and Results}

The appendix below provides basic notions on some concepts used in this paper.

%------------------------------------------------------------------------------
%\subsection{Strongly Typical Sequences and Delta-Convention}
\label{sec:typical}

Following~\cite{csiszar1982information}, we use in this paper \emph{strongly typical sets} and the so-called \emph{Delta-Convention}. 
Some useful facts are recalled here.
Let $X$ and $Y$ be random variables on some finite sets $\cX$ and $\cY$, respectively. We denote by $P_{XY}$ (resp. $P_{Y|X}$, and $P_X$) the joint probability distribution of $(X,Y)$ (resp. conditional distribution of $Y$ given $X$, and marginal distribution of $X$). 

%% - - - - - - - - - - - - - - - - - - - - - - - - - - - - - - - - - - - - - - -
\begin{definition}
For any sequence $x^n\in\cX^n$ and any symbol $a\in\cX$, notation $N(a|x^n)$ stands for the number of occurrences of $a$ in $x^n$.
\end{definition}

\begin{definition}
A sequence $x^n\in\cX^n$ is called \emph{(strongly) $\delta$-typical} w.r.t. $X$ (or simply \emph{typical} if the context is clear) if
\[
\abs{\frac1n N(a|x^n) - P_X(a)} \leq \delta \ \text{ for each } a\in\cX \ ,
\]
and $N(a|x^n)=0$ for each $a\in\cX$ such that $P_X(a)=0$.
The set of all such sequences is denoted by $\typ{X}$.
\end{definition}

\begin{definition}
Let $x^n\in\cX^n$.
A sequence $y^n\in\cY^n$ is called \emph{(strongly) $\delta$-typical (w.r.t. $Y$) given $x^n$} if $ \text{for all } a\in\cX, b\in\cY $
\[
\abs{\frac1n N(a,b|x^n,y^n) - \frac1n N(a|x^n)P_{Y|X}(b|a)} \leq \delta \ ,
\]
and, $N(a,b|x^n,y^n)=0$ for each $a\in\cX$, $b\in\cY$ such that $P_{Y|X}(b|a)=0$.
The set of all such sequences is denoted by $\typ{Y|x^n}$.
\end{definition}

% - - - - - - - - - - - - - - - - - - - - - - - - - - - - - - - - - - - - - - -
\emph{Delta-Convention~\cite{csiszar1982information}:\ }
For any sets $\cX$, $\cY$, there exists a sequence $\{\delta_n\}_{n\in\bN^*}$ such that lemmas below hold.\footnote{As a matter of fact, $\delta_n\to0$ and $\sqrt{n}\,\delta_n\to\infty$ as $n\to\infty$.}
From now on, typical sequences are understood with $\delta=\delta_n$. 
Typical sets are still denoted by $\typ{\cdot}$.

% - - - - - - - - - - - - - - - - - - - - - - - - - - - - - - - - - - - - - - -
\begin{lemma}[{\cite[Lemma~1.2.12]{csiszar1982information}}]\label{lem-concentration}
There exists a sequence $\eta_n\toas{n\to\infty}0$ such that
\[
P_X^n(\typ{X}) \geq 1 - \eta_n \ .
\]
\end{lemma}

%% - - - - - - - - - - - - - - - - - - - - - - - - - - - - - - - - - - - - - - -
%\begin{lemma}[{\cite[Lemma~1.2.13]{csiszar1982information}}]
%\label{lem:cardTyp}
%There exists a sequence $\eta_n\toas{n\to\infty}0$ such that, for each $x^n\in\typ{X}$,
%\begin{IEEEeqnarray*}{c}
%\abs{\frac1n \log \norm{\typ{X}} - H(X)} \leq \eta_n 		\ ,\\
%\abs{\frac1n \log \norm{\typ{Y|x^n}} - H(Y|X)} \leq \eta_n	\ .
%\end{IEEEeqnarray*}
%\end{lemma}

%% - - - - - - - - - - - - - - - - - - - - - - - - - - - - - - - - - - - - - - -
%\begin{lemma}[Asymptotic equipartition property]
%\label{lem:AEP}
%There exists a sequence $\eta_n\toas{n\to\infty}0$ such that, for each $x^n\in\typ{X}$ and each $y^n\in\typ{Y|x^n}$,
%\begin{IEEEeqnarray*}{c}
%\abs{-\frac1n \log P^n_X(x^n) - H(X)} \leq \eta_n 				\ ,\\
%\abs{-\frac1n \log P^n_{Y|X}(y^n|x^n) - H(Y|X)} \leq \eta_n		\ .
%\end{IEEEeqnarray*}
%\end{lemma}

% - - - - - - - - - - - - - - - - - - - - - - - - - - - - - - - - - - - - - - -
\begin{lemma}[Joint typicality lemma~\cite{cover2006elements}]
\label{lem:jointTypicality}
There exists a sequence $\eta_n\toas{n\to\infty}0$ such that for each $x^n\in\typ{X}$:  
\[
\abs{-\frac1n \log P^n_Y(\typ{Y|x^n}) - I(X;Y)} \leq \eta_n \
.
\]
\end{lemma}

%\begin{IEEEproof}
%\begin{IEEEeqnarray*}{rCl}
%P^n_Y(\typ{Y|x^n})
%	&=&						\sum_{y^n\in\typ{Y|x^n}} P^n_Y(y^n)				\\
%	&\stackrel{(a)}{\leq}&	\norm{\typ{Y|x^n}}\,2^{-n[H(Y)-\alpha_n]}		\\
%	&\stackrel{(b)}{\leq}&	2^{n[H(Y|X)+\beta_n]}\,2^{-n[H(Y)-\alpha_n]}	\\
%	&=&						2^{-n[I(X;Y)-\beta_n-\alpha_n]}					\ ,
%\end{IEEEeqnarray*}
%where
%\begin{itemize}
%\item step~$(a)$ follows from the fact that $\typ{Y|x^n}\subset\typ{Y}$ and Lemma~\ref{lem:AEP}, for some sequence $\alpha_n\toas{n\to\infty}0$,
%\item step~$(b)$ from Lemma~\ref{lem:cardTyp}, for some sequence $\beta_n\toas{n\to\infty}0$.
%\end{itemize}
%The reverse inequality $P^n_Y(\typ{Y|x^n})\geq 2^{-n[I(X;Y)+\beta_n+\alpha_n]}$ can be proved following similar argument. 
%\end{IEEEproof}

%------------------------------------------------------------------------------
%\subsection{Csisz\'ar and K\"orner's Identity}
%
\begin{lemma}[{Csisz\'ar \& K\"orner's sum-identity~\cite[Lemma~7]{CsiszarConfidential}}]
Consider two random sequences $X^n$ and $Y^n$, and a constant $C$ (independent of time).
The following identity holds:
\begin{equation}
\sum_{i=1}^n I(Y_{i+1}^n ; X_i | C X^{i-1})
        =       \sum_{i=1}^n I(X^{i-1} ; Y_i | C Y_{i+1}^n)  \ .\label{sec:csiszarkorner}
\end{equation}
\end{lemma}
 
\begin{IEEEproof}
\begin{IEEEeqnarray*}{rCl}
&& \sum_{i=1}^n  \Bigl[ I(Y_{i+1}^n ; X_i | C X^{i-1}) - I(X^{i-1} ; Y_i | C Y_{i+1}^n)  \Bigr] \nonumber  \\
&=& \sum_{i=1}^n \Bigl[ I(Y_{i+1}^n ;X_i X^{i-1} | C ) - I(X^{i-1} ; Y_i Y_{i+1}^n | C)  \Bigr] \\
&=& \sum_{i=1}^n \Bigl[ I(Y_{i+1}^n ; X^i | C ) - I(Y_i^n ; X^{i-1} | C)  \Bigr] \\
&=& \sum_{i=1}^n \Bigl[ S_i - S_{i-1}  \Bigr] \\
&=& S_n - S_0 \\
&=& 0
\end{IEEEeqnarray*}
where: $S_i \triangleq  I(Y_{i+1}^n ; X_i^n | C ) $, and where we define $Y_{n+1}^n = X^{0}= \emptyset$ which leads to $S_n = S_0 = 0$. 
\end{IEEEproof} 
 
%================= SECTION =================

	\section{Proof of Lemma \ref{Inequalities}}\label{ProofLemmaInequalities}
	We want to show the following set of inequalities: 
	\begin{IEEEeqnarray}{rCl}
		 I(Q^n ; Z^n) &\leq& n  \,I(Q;Z) + n \, \epsilon_n  \ , \\
		 I(U_1^n ; U_2^n | Q^n) &\leq& n \, I(U_1; U_2 |Q) + n \, \epsilon_n  \ ,\label{eq-second} \\
		 I(U_1^n  U_2^n; Z^n |Q^n) &\leq& n  \, I(U_1 U_2; Z |Q) + n \, \epsilon_n \ .
		\end{IEEEeqnarray}
		All inequalities can be proved using the same approach, so we only prove inequality~\eqref{eq-second}. 
		
		Let $\cE$ be the indicator function defined by 		
		\begin{equation}
		\cE\triangleq \left\{\begin{array}{lcl}
		1 & & \textrm{ if $(q^n, u_1^n, u_2^n ) \in \typ{QU_1U_2}$} \\ 
		0 & & \textrm{ otherwise} 
		\end{array}\right.
		\end{equation}
with probability  $\mathds{P} (\cE = 1)$. We have that: 		 
		\begin{IEEEeqnarray}{rCl}
	    &&I(U_1^n ; U_2^n | Q^n)  \IEEEnonumber \\ 
	    &\leq&  I(U_1^n, \cE; U_2^n | Q^n)   \\
	    &=& I(U_1^n ; U_2^n | Q^n , \cE) + I( \cE; U_2^n | Q^n)\\
	    &\overset{(a)}{\leq}&  I(U_1^n ; U_2^n | Q^n, \cE) + 1 \\
	    &=& \mathds{P} (\cE = 1)  I(U_1^n ; U_2^n | Q^n , \cE= 1) \IEEEnonumber \\ 
	  &&   \qquad   + \mathds{P} (\cE = 0)  I(U_1^n ; U_2^n | Q^n , \cE= 0) + 1\ \ \  \\
	    &\leq& I(U_1^n ; U_2^n | Q^n , \cE= 1) \IEEEnonumber \\ 
	    && \qquad+  n \, \mathds{P} (\cE = 0)  \log_2 (\|U_2\|) + 1  \ ,  
	\end{IEEEeqnarray}	    
	where $(a)$ is due to upper bounding $h_2(\cE)\leq 1$. By the codebook generation, as $n$ grows large, $\mathds{P} (\cE = 0) $ can be made arbitrarily small. Note that if encoding is  succeeds, only jointly typical sequences $U_1^n $ and $U_2^n$ are sent. %Note that, in the proof of the other two inequalities, it is rather the LLN that allows us to state such a fact about $P_{\cE}(0)$ .
Then, if $\cE = 1$, as a result of Lemma \ref{lem:jointTypicality}, we can have 
		\begin{equation}
		I(U_1^n ; U_2^n | Q^n , \cE= 1) \leq n I(U_1; U_2 | Q) + n \epsilon_n  
		\end{equation}
		and thus, 
		\begin{equation}
		\dfrac{1}{n} I(U_1^n ; U_2^n | Q^n ) \leq  I(U_1; U_2 | Q) + 2 \epsilon_n  \ . 
		\end{equation}
		The remaining inequalities follow in a similar manner and thus details are omitted here.

%================= SECTION =================

	\section{Proof of Lemma \ref{LemmaUpperBound}}\label{ProofLemmaUpperBound}
	 	In this section, we want to prove the following: 
	 	\[
	 	\limsup_{n \rightarrow \infty }\dfrac{1}{n}  H(S_0 S_1 S_2 | Z^n \bar{W}_0 \bar{W}_1 \bar{W}_2 , \mathcal{C} ) \leq \max \left\{0 , I_1, I_2, I_3 , I_4\right\}  \ .
	 	\]
	 	To do this, given the output $z^n$ and the messages $(\bar{W}_0, \bar{W}_1, \bar{W}_2)$, let us define $\mathcal{S}$ as the set of indices $( s_0, s_1,s_2) $ falling in the respective messages' bins, such that: 
	 	\begin{equation}
			(q^n(s_0) , u_1^n(s_0, s_1) ,  u_2^n(s_0, s_2) , z^n ) \in \typ{QU_1U_2Z} \ .  	 	
	 	\end{equation}
	 Then, we can show that the expected size of this list, over all codebooks, is upper bound by 
	 	\begin{equation}\label{IneqSizeList}
	 	\mathds{E}(\|\mathcal{S}\|) \leq 1 + 2^{nI_1} + 2^{nI_2} +  2^{nI_3} +  2^{nI_4}  \ ,
	 	\end{equation}
	 	where: 
	\begin{IEEEeqnarray}{rCl}
	I_1 &=& T_1 - R_1 - I(U_1; Z U_2 | Q)  \ ,\\ 
	I_2 &=& T_2 - R_2 - I(U_2; Z U_1 | Q)  \ ,\\
	I_3 &=& T_1 - R_1 + T_2 - R_2  \IEEEnonumber \\ 
	  &&   \qquad  - I(U_1U_2 ; Z | Q ) - I(U_1; U_2 |Q)  \ ,\\
	I_4 &=& T_0 - R_0 + T_1 - R_1 + T_2 - R_2 \IEEEnonumber \\ 
	  &&   \qquad  - I(QU_1U_2 ; Z ) - I(U_1; U_2 |Q) \ .
	\end{IEEEeqnarray}
	
	To see this, one can note that: 
	\begin{IEEEeqnarray}{rCl}
		\mathds{E} \| \mathcal{S}  \| &=& \mathds{P}\{ (S_0,S_1,S_2) \in \mathcal{S} \} + \sum_{\mathclap{(s_0,s_1,s_2) \neq (S_0,S_1,S_2)}} \mathds{P}\{(s_0,s_1,s_2) \in \mathcal{S}\}  \qquad .
	\end{IEEEeqnarray}	
	 where $(S_0,S_1,S_2 )$ are the true indices chosen by the source. 
	 
	 Due to the LLN and the codebook construction, and Lemma \ref{lem-concentration}, we can show that: 
	 \begin{equation}
       \mathds{P}\{ (S_0,S_1,S_2) \in \mathcal{S} \}  \geq 1 - \eta 
     \end{equation}
	 As for the probability of undetected errors, we can distinguish many cases following the values of $(s_0, s_1, s_2)$.     Hereafter, we give only representative classes of errors. 
	 \begin{itemize}
	\item If $s_1 \neq S_1$ and $(s_0,s_2) = (S_0, S_2)$, then by similar tools to Lemma \ref{lem:jointTypicality}, we can show that: 
    	 \begin{equation}
    	   \mathds{P}\{ (S_0,s_1,S_2) \in \mathcal{S} \} \leq 2^{[-n I(U_1; Z U_2 |Q) + n \epsilon_n]}
         \end{equation}
     \item If $s_1 \neq S_1$, $s_2 \neq S_2$ and $s_0 = S_0$, then: 
     \begin{equation}
    	   \mathds{P}\{ (S_0,s_1,s_2) \in \mathcal{S} \} \leq 2^{[-n I(U_1 U_2; Z |Q) -n I(U_1 ; U_2 |Q)+ n \epsilon_n]}
         \end{equation}
        \item Last, if $s_0 \neq S_0 $, then for all $(s_1,s_2)$, 
        \begin{equation}
    	   \mathds{P}\{ (s_0,s_1,s_2) \in \mathcal{S} \} \leq 2^{[-n I(QU_1U_2 ; Z ) -n I(U_1; U_2 |Q)+ n \epsilon_n]}
         \end{equation}
\end{itemize}
	
	Now, once the list size has been bounded, by defining 	
			\begin{equation}
		\cE\triangleq \left\{\begin{array}{lcl}
		1 & & \textrm{ if $(S_0, S_1, S_2) \in \mathcal{S}$ } \\ 
		0 & & \textrm{ if otherwise} 
		\end{array}\right.
		\end{equation}
we have that
	\begin{IEEEeqnarray}{rCl}
 && H(S_0 S_1 S_2 | Z^n \bar{W}_0 \bar{W}_1 \bar{W}_2, \mathcal{C} ) \IEEEnonumber \\ 
	  &=& I( \cE; S_0 S_1 S_2 | Z^n \bar{W}_0 \bar{W}_1 \bar{W}_2, \mathcal{C})  \nonumber\\
	   \IEEEeqnarraymulticol{3}{r}{+ H(S_0 S_1 S_2 | Z^n \bar{W}_0 \bar{W}_1 \bar{W}_2 ,\cE , \mathcal{C} )} \\
	  & \overset{(a)}{\leq}& 1 +  H(S_0 S_1 S_2 | Z^n \bar{W}_0 \bar{W}_1 \bar{W}_2 ,\cE , \mathcal{C} ) \\
	  & \overset{(b)}{\leq}& 1 +  H(S_0 S_1 S_2 | Z^n \bar{W}_0 \bar{W}_1 \bar{W}_2 ,\cE=1 , \mathcal{C} ) \nonumber\\
	   \IEEEeqnarraymulticol{3}{r}{+  \mathds{P} (\cE = 0) H(S_0 S_1 S_2 | \bar{W}_0 \bar{W}_1 \bar{W}_2 ) \ , } 
	\end{IEEEeqnarray}	
	where $(a)$ comes from that the entropy of the binary variable $\cE$ is upper-bounded by $1$ while $(b)$ follows by upper bounding: $\mathds{P} (\cE = 1) \leq 1 $ and 
	\begin{IEEEeqnarray*}{rCl}
&&H(S_0 S_1 S_2 | Z^n  \bar{W}_0  \bar{W}_1 \bar{W}_2 ,\cE = 0 , \mathcal{C} ) \nonumber \\
&& \quad \leq  H(S_0 S_1 S_2|\bar{W}_0 \bar{W}_1 \bar{W}_2 )  \ .
\end{IEEEeqnarray*}	
By our codebook construction  and Lemma~\ref{lem-concentration}, again $\mathds{P} (\cE = 0)$ can be made arbitrarily small. Next, note that: 
	\begin{IEEEeqnarray}{rCl}
	  && H(S_0S_1S_2 | Z^n \bar{W}_0 \bar{W}_1 \bar{W}_2 ,\cE =1, \mathcal{C} )   \IEEEnonumber \\ 
	    & \overset{(a)}{=}& H(S_0 S_1 S_2 | Z^n \bar{W}_0 \bar{W}_1 \bar{W}_2 ,\cE =1, \mathcal{C}  , \mathcal{S} , \|\mathcal{S}\|) \\
	  & \leq & H(S_0 S_1 S_2 | \cE=1, \mathcal{S} , \|\mathcal{S}\|) \\
	  &=& \sum_{\mathclap{s\in \textrm{supp}(\| \mathcal{S}\|) }} P(\|\mathcal{S}\| = s )H(S_0 S_1 S_2 | \cE=1, \mathcal{S} , \|\mathcal{S}\| = s) \ \  \ \  \ \ \\
	  &\overset{(b)}{\leq}& \sum_{s\in \textrm{supp}(\| \mathcal{S}\|) }  P(\|\mathcal{S}\| = s ) \log_2(s)\\
	   &=& \mathds{E}\left[ \log_2(\|\mathcal{S}\|) \right]\\
	  &\overset{(c)}{\leq}&  \log_2 \left(\mathds{E}  \|\mathcal{S}\| \right)  \\
	  &\overset{(d)}{\leq}& n \max \left\{ 0, I_1, I_2, I_3, I_4 \right\} + \log_2{(5)}  \ ,
	\end{IEEEeqnarray}
	where $(a)$ follows form the fact that $\mathcal{S}$ and $\|\mathcal{S}\|$ are functions of the output $Z^n$, the codebook and the chosen messages to be sent; $(b)$ is a result of that knowing $\cE=1$, the sent indices $(S_0,S_1,S_2) $ belong to the set  $\mathcal{S}$ and thus their uncertainty can not exceed the log cardinality of that set; and finally, $(c)$ is a consequence of Jensen's inequality while $(d)$ comes from \eqref{IneqSizeList} along with an application of the \emph{log-sum-exp} inequality: 
\begin{equation}
  \log_2 \left( \sum_{x \in \cX}2^x \right) \leq \max_{x \in \cX} x +  \log_2( \|  \cX \| ) \ .
\end{equation}	  
This, along with the previous remarks yields the desired inequality: 
		\begin{IEEEeqnarray*}{rCl}
	 &&	\limsup_{n \rightarrow \infty }\dfrac{1}{n}  H(S_0 S_1 S_2 | Z^n \bar{W}_0 \bar{W}_1 \bar{W}_2 ,\mathcal{C} )\nonumber \\
	 	&& \qquad  \leq \max \left\{0 , I_1, I_2, I_3 , I_4 \right\}  \ .
	\end{IEEEeqnarray*}	
		
%% -------------------------------------------------------------------------------------------------------------------%%%
\section{Fourier-Motzkin Elimination}\label{appendix-missing}
 	We resort to FME, recalling all the constraints: 
	 \begin{IEEEeqnarray}{rCl}
		    T_1 &\leq& I(U_1; Y_1 | Q )  \ ,\\
		    T_1 + T_0 &\leq& I(Q U_1; Y_1)  \ ,\\
		    T_2 &\leq& I(U_2; Y_2 | Q)  \ ,\\
		    T_2 + T_0 &\leq& I(Q U_2; Y_2)  \ ,\\
		    T_0 - \bar{R}_0 &\geq& I(Q; Z)  \ ,\\
		    T_0 - \bar{R}_0 + T_1 - \bar{R}_1  &\geq& I(QU_1;Z) \ , \\
		    T_0 - \bar{R}_0 + T_2 - \bar{R}_2  &\geq&  I(QU_2;Z)  \ ,\\
		    T_0 + T_1 + T_2 - ( \bar{R}_0 + \bar{R}_1 + \bar{R}_2 ) &\geq&   I(QU_1U_2  ; Z) \IEEEnonumber \\ 
	  &&     + I(U_1; U_2 |Q)  \ , \\ 
		    T_1 - \bar{R}_1 - \tilde{R}_1 + T_2 - \bar{R}_2 - \tilde{R}_2  &\geq& I(U_1; U_2 |Q)  \ ,\\
	 		0 \leq \tilde{R}_1 \leq T_1-\bar{R}_1  \! &,& \!  0 \leq \tilde{R}_2 \leq T_2-\bar{R}_2  \ . \nonumber 
		\end{IEEEeqnarray}
The resulting rate region after FME is as follows: 
	 \begin{IEEEeqnarray}{rCl}
		    \bar{R}_1 &\leq& I(U_1; Y_1 | Q)  \ ,\\
		    \bar{R}_1 + \bar{R}_0 &\leq& I(Q U_1; Y_1)  - I(QU_1;Z) \ ,\\
		    \bar{R}_2 &\leq& I(U_2; Y_2 | Q)  \ ,\\
		    \bar{R}_2 + \bar{R}_0 &\leq& I(Q U_2; Y_2) -I(QU_2;Z)  \ ,\\
		    \bar{R}_1 + \bar{R}_2 &\leq& I(U_1; Y_1 | Q) +I(U_2; Y_2 | Q)  \IEEEnonumber \\ 
	  &&   \quad - I(U_1; U_2 |Q)  \ , \\
		    \bar{R}_0 + \bar{R}_1 + \bar{R}_2 &\leq& I(QU_1; Y_1) +I(U_2; Y_2 | Q)  \IEEEnonumber \\ 
	  &&   \quad - I(QU_1U_2  ; Z) -  I(U_1; U_2 |Q) \ , \\
		    \bar{R}_0 + \bar{R}_1 + \bar{R}_2 &\leq& I(QU_2; Y_2) +I(U_1; Y_1 | Q)  \IEEEnonumber \\ 
	  &&   \quad - I(QU_1U_2  ; Z) -  I(U_1; U_2 |Q) \ , \\
2\, \bar{R}_0 + \bar{R}_1 + \bar{R}_2 &\leq& I(QU_2; Y_2) +I(QU_1; Y_1 )-I(QU_1U_2  ; Z)\nonumber\\
& &  -  I(U_1; U_2 |Q) - I(Q; Z)  \ .
		\end{IEEEeqnarray}
	%% -------------------------------------------------------------------------------------------------------------------%%%
\subsection*{Eliminating rate splitting parameters:}
	
The achievable rate region writes then as: 
	 \begin{IEEEeqnarray}{rCl}
		    R_1 - R_{01} &\leq& I(U_1; Y_1 | Q)  \ ,\\
		    R_1 +R_{02}&\leq& I(Q U_1; Y_1) - I(QU_1;Z) \ ,\\
		    R_2  - R_{02} &\leq& I(U_2; Y_2 | Q)  \ ,\\
		    R_2 +R_{01}&\leq& I(Q U_2; Y_2) - I(QU_2;Z) \ ,\\
		    R_1 - R_{01} + R_2 - R_{02}  &\leq& I(U_1; Y_1 | Q) +I(U_2; Y_2 | Q)  \IEEEnonumber \\ 
	  &&    - I(U_1; U_2 |Q)  \ , \\
		   R_1 + R_2 &\leq& I(QU_1; Y_1) +I(U_2; Y_2 | Q) \IEEEnonumber \\ 
	  &&      - I(QU_1U_2  ; Z) -  I(U_1; U_2 |Q)  \ ,\\
		    R_1 + R_2 &\leq& I(QU_2; Y_2) +I(U_1; Y_1 | Q)  \IEEEnonumber \\ 
	  &&    - I(QU_1U_2  ; Z) -  I(U_1; U_2 |Q)  \ ,\\
		    R_1 + R_2 + R_{01} + R_{02} &\leq& I(QU_2; Y_2) +I(QU_1; Y_1 ) - I(Q; Z) \nonumber\\
		   &  &-I(QU_1U_2 ; Z) - I(U_1; U_2 |Q)  \ .		
		   \end{IEEEeqnarray}
Eliminating the rates splitting parameters $R_{01}$ and $R_{02}$ with the positivity constrains: $R_{0,j} > 0 $ and $R_j - R_{0j} > 0  $ for $j \in \{1,2\}$,  yields the desired inner bound. 
%================= SECTION =================
	 	
		\section{Proof of Lemma \ref{LemmaConvexity} }\label{ProofLemmaConvexity}
In this section, we show the convexity of the rate region given by: 
\begin{equation}
\mathcal{R}  \,: \left\{ \begin{array}{rcl}
R_1 &\leq& (1-e) \, h_2(x) + h_2(p) - h_2(p*x) \ , \\
R_2 &\leq&  h_2(p*x) - h_2(p_2*x) \ ,
\end{array}\right.  
\end{equation}
where the union is over $ x \in [0:0.5]$. 

Obtaining this result comes to writing an equivalent of Mrs. Gerber's Lemma~\cite{Wyner1973} in the presence of an eavesdropper in the same fashion as in~\cite{Wyner1973}. Our aim will be to show that, for the corner point of this region, the rate $R_2$ is a concave function of the rate $R_1$. 

Let us define the function $f_1$ as follows:  
\begin{equation}
R_1 = f_1(x) \triangleq (1-e) h_2(x)+ h_2(p) - h_2(p * x) \ . 
\end{equation}
We have that: 
 \begin{equation}
 f^\prime_1(x) = (1-e) h^\prime_2(x)  + (1 - 2p) h^\prime_2(p*x)  \ , 
\end{equation}
and, 
 \begin{equation}
 f^{\prime\prime}_1(x) = (1-e) h^{\prime\prime}_2(x)  + (1 - 2p)^2 h^{\prime\prime}_2(p*x) \ , 
\end{equation}
where:
\begin{equation}
  h^{\prime}_2(x)  = \log_2 \left( \dfrac{1-x}{x}\right) \qquad \text{and}\qquad  h^{\prime\prime}_2(x) = - \dfrac{1}{x(1-x)}  \ .  
\end{equation} 
Let us also define the function $f_2$ as: 
\begin{equation}
R_2 = f_2(x) \triangleq h_2(p_2*x) - h_2(p * x) \ . 
\end{equation}
In the same fashion, we can write: 
 \begin{equation}
 f^\prime_2(x) = (1 - 2p_2)   h^\prime_2(p_2*x)  - (1 - 2p)  h^\prime_2(p*x)   \ ,
\end{equation}
and  
\begin{equation}
 f^{\prime\prime}_2(x) =  (1 - 2p_2)^2 h^{\prime\prime}_2(p_2*x)  -  (1 - 2p)^2 h^{\prime\prime}_2(p*x) \ .
\end{equation}
To show that:  
 \begin{equation}
 \frac{d^2 R_2}{d R_1^2} = \frac{d^2 f_2}{d f_1^2} \leq 0 \ ,
\end{equation}
we observe that: 
\begin{IEEEeqnarray}{rCl}
\frac{d f_2}{d f_1 } &=& \frac{d f_2}{d x} \,\frac{d x}{d f_1 } =  \frac{d f_2}{d x} \, \frac{d f_1^{-1}(y)}{d y }  \nonumber \\
&=& \frac{1}{f^\prime_1 (f_1^{-1}(y))} \,\frac{d f_2}{d x}  
= \frac{1}{f^\prime_1 (x)} \, \frac{d f_2}{d x}  \ .
\end{IEEEeqnarray}
As such, one can write in the same manner that: 
\begin{equation} 
\frac{d^2 f_2}{d f_1 ^2 }  =  \dfrac{f^{\prime\prime}_2(x) f^\prime_1(x) - f^{\prime\prime}_1(x) f^\prime_2(x)}{\left(f^\prime_1(x)\right)^3} \ .
\end{equation}
Since $ 0\leq x\leq \frac{1}{2}$ , then $0\leq p*x \leq \frac{1}{2}$, and thus, one can easily check that: 
\begin{equation} 
f^\prime_1(x) \geq 0 \ .
\end{equation}
Thus, it suffices to show that for all $ x \in [0:0.5]$, 
\begin{equation}  \label{Desired}
f^{\prime\prime}_2(x) f^\prime_1(x) - f^{\prime\prime}_1(x) f^\prime_2(x)  \leq 0   \ .
\end{equation}
For notation convenience, we let: 
\begin{equation}
  a \triangleq   1 - 2 p  \qquad \text{and}\qquad  a_2  \triangleq  1 - 2 p_2 \ .
\end{equation} 
Now, one can write that: 
\begin{IEEEeqnarray}{rCl}
&&  f^{\prime\prime}_2(x) f^\prime_1(x) - f^{\prime\prime}_1(x) f^\prime_2(x)   \IEEEnonumber \\ 
  &=&    a^2 \, h_2^{\prime \prime}(p*x) \, \Bigl[  (1-e)\, h_2^\prime(x) - a_2 \, h_2^\prime (p_2*x)  \Bigr]  \nonumber\\
 \IEEEeqnarraymulticol{3}{r}{ -  a_2^2 \, h_2^{\prime \prime}(p_2*x) \, \Bigl[  (1-e)\,  h_2^\prime(x) - a \, h_2^\prime (p *x)  \Bigr]} \nonumber \\
 \IEEEeqnarraymulticol{3}{r}{-  (1-e) \, h_2^{\prime \prime}(x) \, \Bigl[  a\, h_2^\prime(p*x) - a_2 \, h_2^\prime (p_2 *x)  \Bigr]\ , }   
\end{IEEEeqnarray}
and thus 

\begin{IEEEeqnarray}{rCl}
&&   \dfrac{f^{\prime\prime}_2(x) f^\prime_1(x) - f^{\prime\prime}_1(x) f^\prime_2(x) }{h_2^{\prime \prime}(p*x)  \, h_2^{\prime \prime}(p_2*x) \, h_2^{\prime \prime}(x)}  \IEEEnonumber \\ 
  &=&  a^2 \, \dfrac{ (1-e)\, h_2^\prime(x) - a_2 \, h_2^\prime (p_2*x)  }{ h_2^{\prime \prime}(p_2*x) \, h_2^{\prime \prime}(x)}  \IEEEnonumber \\ 
	  &&   \qquad  - \,  a_2^2 \, \dfrac{(1-e)\,  h_2^\prime(x) - a \, h_2^\prime (p *x) }{ h_2^{\prime \prime}(p*x) \, h_2^{\prime \prime}(x) }  \nonumber\\ 
&&  \qquad \qquad - \, (1-e) \, \dfrac{a\, h_2^\prime(p*x) - a_2 \, h_2^\prime (p_2 *x)}{h_2^{\prime \prime}(p_2*x)h_2^{\prime \prime}(p*x) }   \ .\qquad
\end{IEEEeqnarray}

Let us now define a variable $\alpha$ such that:  $\alpha  \triangleq  1 - 2x $. 
We have that: 
 \begin{equation}
  a \cdot \alpha = 1 - 2 (p*x)  \qquad \text{and}\qquad  a_2 \cdot \alpha = 1 - 2 (p_2*x) \ .
\end{equation} 
Moreover: 
\begin{IEEEeqnarray}{rCl}
h_2^\prime (x) &=& \log_2 \left(\dfrac{1-x}{x} \right)=  \log_2 \left(\dfrac{1+ \alpha}{1- \alpha} \right)  \ , \\
h_2^{\prime\prime} (x) &=& - \dfrac{1}{x(1-x)} = - \dfrac{4}{1 - \alpha^2} \ .
\end{IEEEeqnarray}
 Then, to show the desired inequality \eqref{Desired}, since:
 \begin{equation}
 h_2^{\prime \prime}(p*x)  \, h_2^{\prime \prime}(p_2*x) \, h_2^{\prime \prime}(x)  \leq 0 \ , 
 \end{equation}
 one only has to show, after some simplifications, that: 
 \begin{IEEEeqnarray}{rCl}
	 &&   - a_2 \, \left(1 - (a_2 \alpha)^2 \right)\left[a^2 - 1 + e \left (1- (a \alpha )^2 \right) \right]\log_2 \left(\dfrac{1+ a_2\alpha}{1- a_2\alpha} \right)  \nonumber \\ 
	  &&  +  a \, \left(1 - (a \alpha)^2 \right)\left[a_2^2 - 1 + e \left (1- (a_2 \alpha )^2 \right) \right]\log_2 \left(\dfrac{1+ a\alpha}{1- a\alpha} \right)   \nonumber   \\ 
	  &&  (1-e) \, \left(a^2 - a_2^2\right)  \log_2 \left(\dfrac{1+ \alpha}{1- \alpha} \right)   \geq 0 \ . \label{Final}
\end{IEEEeqnarray}
We will resort to the known series expansion of the log: 
\begin{equation}
   \log\left(\dfrac{1+\alpha}{1-\alpha}\right) = 2 \sum^\infty_{\substack{k=1 \\ \textrm{$k$ odd}}}\dfrac{\alpha^k}{k} \ ,
\end{equation} 
to write that the inequality \eqref{Final}, after simplifications, requires: 

 \begin{IEEEeqnarray}{rCl}
 &(a_2^2 - a^2)& \sum^\infty_{\substack{k=5 \\ \textrm{$k$ odd} }} \alpha^k \biggl[ \left( \dfrac{1}{k-2} -  \dfrac{1}{k}\right) T_k  -   \left( \dfrac{1}{k-4} -  \dfrac{1}{k-2}\right) V_k  \biggr] \nonumber \\
  &\geq & - \dfrac{2}{3} \alpha^3 T_3 \ ,
\end{IEEEeqnarray}
for all $\alpha \in[0:1]$,  where
 \begin{IEEEeqnarray}{rCl}
  T_k &=& (1-e) \left(  1 - \dfrac{ a_2^{k+1} - a^{k+1}}{ a_2^2 - a^2 }\right) + a_2^2 \, a^2 \, \dfrac{ a_2^{k-1} - a^{k-1}}{ a_2^2 - a^2 }   \qquad  \\
  V_k &=& e \, a_2^2 \, a_2 \, \dfrac{ a_2^{k-3} - a^{k-3} }{ a_2^2 - a^2  } \ .
\end{IEEEeqnarray}
By hypothesis $p_2 \leq p$ and hence $a_2^2 - a^2 \geq 0$. We are thus left with only the analysis of the summation. In the sequel, we show the following results on summation operand.

\begin{lemma}[Properties of some series]\label{LemmaBECBSC}
\ 
\begin{enumerate}
\item The sequence $(T_k)_k$ dominates the sequence $(V_k)_k$ in that:
\begin{equation}
( \forall k \in \mathds{N}_{\textrm{odd}} )  \  ,\    T_k \geq V_k \geq 0  \ .
\end{equation}
\item If  $a^2 + a_2^2 \leq 1$, then $(V_k)_{\substack{k\geq 5}}$ for  $k$ odd is a decreasing sequence. 
\item The following identity holds: 
\begin{equation}
 - \dfrac{2}{3} \alpha^3 T_3 = \sum^\infty_{\substack{k=5 \\ \textrm{$k$ odd}}} \alpha^3 T_3 \left( \dfrac{1}{k-2} -  \dfrac{1}{k} -  \dfrac{1}{k-4} + \dfrac{1}{k-2} \right)  \ .
\end{equation}
\end{enumerate}
\end{lemma}
\begin{IEEEproof} Proof  is given in Appendix~\ref{ProofLemmaBECBSC}. 
\end{IEEEproof}

Indeed, Lemma~\ref{LemmaBECBSC} motivates our choice  $a^2 + a_2^2 \leq 1 $ in the sequel and thus allows us to write: 

\begin{IEEEeqnarray}{rCl}
  && \sum^\infty_{\substack{k=5 \\ \textrm{$k$ odd}}} \alpha^k \biggl[ \left( \dfrac{1}{k-2} - \dfrac{1}{k}\right) T_k  - \left( \dfrac{1}{k-4} -  \dfrac{1}{k-2}\right) V_k  \biggr]  \IEEEnonumber \\ 
	  &&   \qquad + \dfrac{2}{3} \alpha^3 T_3 \\ 
   &\overset{(a)}{=}&  \sum^\infty_{\substack{k=5 \\ \textrm{$k$ odd}}}  \left( \dfrac{1}{k-2} - \dfrac{1}{k}\right) ( \alpha^k T_k - \alpha^3 T_3 ) \IEEEnonumber \\ 
	  &&   \qquad - \left( \dfrac{1}{k-4} - \dfrac{1}{k-2}\right) ( \alpha^k V_k - \alpha^3 T_3 ) \\
 &\overset{(b)}{\geq}& \sum^\infty_{\substack{k=5 \\ \textrm{$k$ odd}}}  \left( \dfrac{1}{k-2} - \dfrac{1}{k}\right) ( \alpha^k V_k - \alpha^3 T_3 )\IEEEnonumber \\ 
	  &&   \qquad  - \left( \dfrac{1}{k-4} - \dfrac{1}{k-2}\right) ( \alpha^k V_k - \alpha^3 T_3 ) \\
   &=& \sum^\infty_{\substack{k=5 \\ \textrm{$k$ odd}}}  \left( \dfrac{1}{k-2} - \dfrac{1}{k} - \dfrac{1}{k-4} + \dfrac{1}{k-2}\right) ( \alpha^k V_k - \alpha^3 T_3 ) \nonumber \\
   &\overset{(c)}{\geq}& 0  \ ,
\end{IEEEeqnarray}
where $(a)$ comes from claim (3) in Lemma~\ref{LemmaBECBSC} and $(b)$ results from claim (1) in Lemma~\ref{LemmaBECBSC} while $(c)$ comes from the fact that
\begin{equation}
\dfrac{1}{k-2} - \dfrac{1}{k} - \dfrac{1}{k-4} + \dfrac{1}{k-2} \leq 0  \ ,
\end{equation}
and hence, since $(V_k)_{k\geq 5}$ is a decreasing sequence, then for all $\alpha \in [0:1]$ we can write that: 
\begin{equation}
( \forall k \geq 5 ) \qquad \alpha^k V_k \leq \alpha^k V_5 \leq \alpha^3 V_5  \ ,
\end{equation}
and since: 
\begin{equation} 
T_3 - V_5 = (1-e) (1 - a^2 - a_2^2 + a^2 \, a_2^2 ) \geq 0  \ ,
\end{equation}
then, 
\begin{equation}
( \forall k \geq 5 ) \qquad \alpha^k V_k  - \alpha^3 T_3 \leq 0 \ .
\end{equation}
It is worth mentioning that the assumption $a_2^2+ a^2 \leq 1$ was used only in the monotony of the sequence $(V_k)$.

%================= SECTION =================

\section{Proof of Lemma \ref{LemmaBECBSC}}\label{ProofLemmaBECBSC}
In this section, we prove the claims stated in Lemma~\ref{LemmaBECBSC}. We start by showing claim (1) which consists to show that $\forall k \in \mathds{N}_{\textrm{odd}}$, $T_k \geq V_k \geq 0$. Let the sequence $(S_k)_{k \in  \mathds{N}_{\textrm{odd}}} $ defined as follows:

\begin{equation}
S_k \triangleq \dfrac{a_2^{k-1} - a^{k-1}}{a_2^2 - a^2}  \ ,
\end{equation}
with  $k-1 \triangleq 2 s $, then one can write that for all $k\geq 3$, 

\begin{equation} \label{SumDef}
S_k =  \sum_{j= 0}^{s-1}  a_2^{2\,j} a^{2\,(s-1-j)} \ .
\end{equation}
Now, we know that:
 \begin{IEEEeqnarray}{rCl}
  T_k &=& (1-e) \left(  1 -  S_{k+2} \right) + a_2^2 \, a^2 \,  S_k \ , \\
  V_k &=& e \, a_2^2 \, a_2 \,  S_{k-2} \ .
\end{IEEEeqnarray}
Let $ k \geq 3 $ for which we have that: 
\begin{equation} \label{EquationDiff}
  T_k  - V_k =  (1-e) \left(  1 -  S_{k+2} \right) + a_2^2 \, a^2 \,  ( S_k - e S_{k-2} )  \ .
\end{equation}
It is easy to check that: 
 \begin{IEEEeqnarray}{rCl}
    S_{k } &=& a^{k-3} + a_2^2 S_{k-2}  \ , \\
    S_{k+2} &=& a_2^{k-1} + a^{k-1} + a^2 \, a_2^2 S_{k-2} \ .
\end{IEEEeqnarray}
Thus, by  substituting these expressions in~\eqref{EquationDiff}, we end up with the next equality: 
 \begin{IEEEeqnarray}{rCl} 
  T_k - V_k &=&  (1-e) \left( 1 - a_2^{k-1} - a^{k-1} \right) \IEEEnonumber \\ 
	  &&   \qquad  + a_2^2 \, a^2 \,  ( S_k - S_{k-2} ) \\
  &=& (1-e) \left( 1 - a_2^{k-1} - a^{k-1} \right) \IEEEnonumber \\ 
	  &&   \qquad + a_2^2 \, a^2 \, ( a^{k-3} + ( a_2^2  -1) \,  S_{k-2} ) \ .
\end{IEEEeqnarray}
Now, from the choice of the system parameters~\eqref{ConstraintBECBSC}, we see that: 
 \begin{equation}
 \max\{ a, a_2^2 \} \leq 1-e \leq a_2 \ .
\end{equation}
Then, to lower bound $T_k - V_k $ we split into the following cases: 
 
(i) If  $1 - a_2^{k-1} - a^{k-1} \geq 0 $, then  
 \begin{IEEEeqnarray*}{rCl} 
 && T_k - V_k  \IEEEnonumber \\ 
	 &=& (1-e) \left( 1 - a_2^{k-1} - a^{k-1} \right)  \nonumber  \\
	 && \qquad \qquad   + a_2^2 \, a^2 \, ( a^{k-3} + ( a_2^2  -1) \,  S_{k-2} )  \nonumber \\
  &\geq& a_2^2 \left( 1 - a_2^{k-1} - a^{k-1} \right) + a_2^2 \, a^2 \, ( a^{k-3} + ( a_2^2  -1) \,  S_{k-2} ) \\
  &=&  a_2^2 \left( 1 - a_2^{k-1} + ( a_2^2 -1) \, a^2 \, S_{k-2} \right)  \\
  &=&  a_2^2 ( 1- a_2^2 ) \left( \dfrac{1 - a_2^{k-1}}{1- a_2^2} -  a^2 \, S_{k-2} \right) \\
  &\overset{(a)}{=}& a_2^2 ( 1- a_2^2 )  \left( \sum_{j= 0}^{s-1}  a_2^{2\,j} - a^2 \sum_{j=0}^{s-2}  a_2^{2\,j} a^{2\,(s-2-j)} \right) \\
  &=&  a_2^2 ( 1- a_2^2 )  \left(   \sum_{j= 0}^{s-1}  a_2^{2\,j} -  \sum_{j= 0}^{s-2}  a_2^{2\,j} a^{2\,(s-1-j)} \right)\\
  &=&  a_2^2 ( 1- a_2^2 )  \left(   a_2^{k-3} + \sum_{j= 0}^{s-2}  a_2^{2\,j} \underbrace{\left( 1 -  a^{2\,(s-1-j)} \right)}_{\geq0}  \right) \\
  &\geq& 0  \ .
\end{IEEEeqnarray*}
where $(a)$ comes from \eqref{SumDef} and some standard manipulations of multinomial coefficients. 

(ii) If $1 - a_2^{k-1} - a^{k-1} \leq 0 $, then
\begin{IEEEeqnarray*}{rCl} 
 &&  T_k - V_k \nonumber \\
   &=&  (1-e) \left( 1 - a_2^{k-1} - a^{k-1} \right) + a_2^2 \, a^2 \, ( a^{k-3} + ( a_2^2  -1) \,  S_{k-2} ) \nonumber \\
  &\geq& \left( 1 - a_2^{k-1} - a^{k-1} \right) + a_2^2 \, a^2 \, ( a^{k-3} + ( a_2^2  -1) \,  S_{k-2} ) \\ 
  &=& 1 - a_2^{k-1} -  a^{k-1} ( 1 - a_2^2 ) - a_2^2 \, a^2 ( 1 - a_2^2 ) \,  S_{k-2} \\
  &=& ( 1- a_2^2 ) \left( \dfrac{ 1 - a_2^{k-1} }{ 1 - a_2^2 } -  a^{k-1} - a_2^2 \, a^2 \,  S_{k-2} \right) \\
  &\overset{(a)}{\geq}& ( 1- a_2^2 ) \left( \dfrac{ 1 - a_2^{k-1} }{ 1 - a_2^2 } -  a^{k-1} - a_2^4  S_{k-2} \right) \\
  &=& ( 1- a_2^2 ) \left(  \sum_{j= 0}^{s-1}  a_2^{2\,j}  -  a^{k-1}  -  a_2^4 \sum_{j=0}^{s-2}  a_2^{2\,j} a^{2\,(s-2-j)}  \right) \\
  &=& ( 1- a_2^2 ) \left(  \sum_{j= 0}^{s-1}  a_2^{2\,j}  -  a^{k-1}  -  \sum_{j=0}^{s-2}  a_2^{2\,(j+2)} a^{2\,(s-2-j)}  \right) \\
  &=& ( 1- a_2^2 ) \left(  \sum_{j= 0}^{s-1}  a_2^{2\,j}  -  a^{k-1}  -  \sum_{j=2}^{s}  a_2^{2\,j} a^{2\,(s-j)}  \right) \\
  &=& ( 1- a_2^2 ) \left(   \underbrace{1 -  a^{k-1}}_{\geq 0} +   \underbrace{a_2^2 -  a_2^{2\,s } }_{\geq 0}+ \sum_{j=2}^{s-1}  a_2^{2\,j} \left(1 -a^{2\,(s-j)} \right)  \right) \nonumber  \\
  &\geq& 0   \ ,
\end{IEEEeqnarray*}
where $(a)$ comes from that $a_2 \geq a \geq 0$.  

This proves our claim. Next, we show that if $a^2 +a_2^2 \leq 1$ then ${(V_k)}_{k\geq 5}$ is decreasing for $k$ odd. Let $k$ be an odd integer such that $k \geq 5$. We have that: 
\begin{IEEEeqnarray}{rCl} 
  \dfrac{V_{k+2} - V_k }{e a^2a_2^2}= S_{k+2} - S_k   \ .
\end{IEEEeqnarray}
We check our last claim by induction, i.e., assuming $S_7 - S_5 \leq 0$ and  
\begin{equation*}
\forall k \geq 5 \ , \  S_{k+2} - S_k \leq 0 \qquad \textrm{ then } \qquad  S_{k+4} - S_{k+2} \leq 0 \ .
\end{equation*}
To this end, we have that: 
\begin{equation}
 S_7 - S_5 = a_2^2 (a^2 + a_2^2 -1) \leq 0 \ .
\end{equation}
Let then $k \geq 5$, such that $ S_{k+2} - S_k  \leq 0$, thus:
\begin{IEEEeqnarray}{rCl} 
 && S_{k+4} - S_{k+2}  \nonumber \\
 &=& a_2^{k+1} + (a^2 - 1) S_{k+2}\\
  &=& a_2^{k+1} + (a^2 - 1)\left( a_2^{k-1} + a^2 S_k \right)  \\
  &=& a_2^{k+1}  - a_2^{k-1}  +  a^2 \left( a_2^{k-1} + (a^2 -1) S_k \right) \\
  &=& \underbrace{ a_2^{k+1}  - a_2^{k-1} }_{\leq 0} +  a^2 \, \underbrace{ \left( S_{k+2} - S_k \right)}_{\leq 0} \\
  &\leq& 0 \ ,
\end{IEEEeqnarray}
which proves the claim. Finally, it is easy to verify that: 

\begin{equation}
 - \dfrac{2}{3} \alpha^3 T_3 = \sum^\infty_{\substack{k=5 \\ \textrm{$k$ odd}}} \alpha^3 T_3 \left( \dfrac{1}{k-2} -  \dfrac{1}{k} -  \dfrac{1}{k-4} + \dfrac{1}{k-2} \right)  \ ,
\end{equation}
by noticing 

\begin{equation}
 \sum^\infty_{\substack{k=5 \\ \textrm{$k$ odd}}}  \left( \dfrac{1}{k-2} -  \dfrac{1}{k} -  \dfrac{1}{k-4} + \dfrac{1}{k-2} \right)  = -\dfrac{2}{3} \ .
\end{equation} 

%================= SECTION =================

\section{Proof of Theorem \ref{ProductInverselyLN}}\label{ProofProduct}
	In this section, we prove the result on the product of the two inversely less-noisy BC with a more-noisy eavesdropper. 
	
	 \subsection{Proof of the achievability}
The achievability easily follows by evaluating the region: 
\begin{equation} 
 \left\{\begin{array}{rcl}
        R_1 &\leq& I( Q U_1; \mathbf{Y} ) - I( Q U_1;\mathbf{Z}  )   \ , \\
		R_2 &\leq& I( Q U_2; \mathbf{T} ) - I( Q U_2; \mathbf{Z}  )  \ , \\
		R_1 + R_2 &\leq&  I(U_1;\mathbf{Y}|Q) + I(QU_2;\mathbf{T} )  \nonumber \\
   &=& \quad- I(QU_1U_2;\mathbf{Z} ) - I(U_1;U_2|Q ) \ ,  \\
		R_1 + R_2 &\leq&  I(QU_1;\mathbf{Y} ) + I(U_2;\mathbf{T} |Q)  \nonumber \\
   &=& \quad- I(QU_1U_2;\mathbf{Z}) - I(U_1;U_2|Q)  \ , \\
		R_1 + R_2 &\leq& I( Q U_1; \mathbf{Y}  ) - I( Q U_1; \mathbf{Z} )+ I( Q U_2; \mathbf{T} )   \nonumber \\
   &=& \quad - I( Q U_2; \mathbf{Z}  )- I(U_1;U_2|\mathbf{Z}Q) \ ,    
 \end{array} \right. 
\end{equation}
based on the choices: $Q = (U_1 , U_2)$ and $U_1 = X_1$ and $U_2 = X_2$ such that  $P_{U_1X_1U_2X_2} = P_{U_1X_1}P_{U_2X_2} $. 

The single rate constraints write thus as: 
\begin{IEEEeqnarray}{rCl}  
	    R_1 &\leq&  I(X_1; Y_1) - I(X_1;Z_1) + I(U_2;Y_2) - I(U_2;Z_2)  \qquad  \\
	    &\overset{(a)}{=}&  I(X_1; Y_1|Z_1) + I(U_2;Y_2) - I(U_2;Z_2) \ , 
\end{IEEEeqnarray}
where $(a)$ is a result of that $Z_1$ is degraded towards $Y_1$.  
The sum-rates follow in a similar fashion, however the last sum-rate is redundant since: 
\begin{equation}
 I(X_1; X_2 | Z_1 Z_2 U_1 U_2) \leq I(X_1; X_2 | U_1 U_2) = 0 \ . 
\end{equation}

	\subsection{Proof of the converse}
Let us concatenate the two outputs $\mathbf{Y} = (Y_1, Y_2)$, $\mathbf{Z}= (Z_1,Z_2)$ and $\mathbf{T}= (T_1,T_2)$.  We start by single rate constraints. 

%% -------------------------------------------------------------------------------------------------------------------%%%
\subsubsection{Single-rate constraints}

By Fano's inequality and the secrecy constraint, we have that: 
\begin{IEEEeqnarray}{rCl}
n(R_1 - \epsilon_n) &\leq& I(W_1;\mathbf{Y}^n) - I(W_1;\mathbf{Z}^n)  \\
&\leq& I(W_1;\mathbf{Y}^n Z_1^n) - I(W_1;\mathbf{Z}^n)  \\
&=& I(W_1;\mathbf{Y}^n |Z_1^n) - I(W_1;Z_2^n| Z_1^n) \ . 
\end{IEEEeqnarray}
Thus, by standard Csisz\'ar \& K\"orner's sum-identity~\eqref{sec:csiszarkorner} and some basic manipulations, we get that: 

\begin{IEEEeqnarray}{rCl}
&& n(R_1 - \epsilon_n) \\
&\leq&  \sum^n_{i=1} \left[ I(W_1; \mathbf{Y}_i | Z_1^n \mathbf{Y}^{i-1} Z^n_{2,i+1}) \right. \nonumber \\
&& \left. \qquad  - I(W_1; Z_{2,i} | Z_1^n \mathbf{Y}^{i-1} Z^n_{2,i+1})\right] \\
&=& \sum^n_{i=1}\left[ I(W_1; Y_{1,i} Y_{2,i} | Z_1^n \mathbf{Y}^{i-1} Z^n_{2,i+1})   \right. \nonumber \\
&& \left. \qquad - I(W_1; Z_{2,i} | Z_1^n \mathbf{Y}^{i-1} Z^n_{2,i+1})\right] \\
&=& \sum^n_{i=1}\left[ I(W_1; Y_{2,i} | Z_1^n \mathbf{Y}^{i-1} Z^n_{2,i+1})  \right. \nonumber \\
&& \left. \qquad  - I(W_1; Z_{2,i} | Z_1^n \mathbf{Y}^{i-1} Z^n_{2,i+1})\right. \nonumber \\
&&   \vphantom{\sum^n_{i=1}} \left.   \qquad  + I(W_1; Y_{1,i} | Y_{2,i} Z_1^n \mathbf{Y}^{i-1} Z^n_{2,i+1}) \right]   \\
&\overset{(a)}{=}& \sum^n_{i=1}\left[ I(W_1; Y_{2,i} | Z_{1,i}^n \mathbf{Y}^{i-1} Z^n_{2,i+1}) \right. \nonumber \\
&& \left. \qquad  - I(W_1; Z_{2,i} | Z_{1,i}^n \mathbf{Y}^{i-1} Z^n_{2,i+1})\right. \nonumber \\
&&   \vphantom{\sum^n_{i=1}}  \left.   \qquad  + I(W_1; Y_{1,i} | Y_{2,i} Z_{1,i}^n \mathbf{Y}^{i-1} Z^n_{2,i+1}) \right]  \\
&\overset{(b)}{\leq}& \sum^n_{i=1} \left[ I(W_1; Y_{2,i} | Z_{1,i}^n \mathbf{Y}^{i-1} Z^n_{2,i+1}) \right. \nonumber \\
&& \left. \qquad - I(W_1; Z_{2,i} | Z_{1,i}^n \mathbf{Y}^{i-1} Z^n_{2,i+1})\right. \nonumber \\
&&  \vphantom{\sum^n_{i=1}} \left. \qquad  + I(X_{1,i}; Y_{1,i} | Z_{1,i} ) \right] \ , 
\end{IEEEeqnarray}
where $(a)$ follows from that $Z_1$ is degraded respect to $Y_1$ and $(b)$ comes from the Markov chain:  
\begin{equation}
 (Z_1^{i-1},\mathbf{Y}^{i-1},Z^n_{2,i+1} ,Y_{2,i}) \mkv  X_{1,i}  \mkv  (Y_{1,i}, Z_{1,i})  \ . 
\end{equation}
Thus, letting $U_{2,i} = W_1$ and $V_2 = (Z_{1,i}^n, \mathbf{Y}^{i-1}, Z^n_{2,i+1} )$ we can simply get the rate constraint: 
\begin{equation}
R_1 \leq I(X_1; Y_1 | Z_1) + I(U_2; Y_2 |V_2) - I(U_2; Z_2 | V_2) \ . 
\end{equation}

%% -------------------------------------------------------------------------------------------------------------------%%%
 \subsubsection{Sum-rate constraint} ~\\ 
We start by writing: 
\begin{IEEEeqnarray}{rCl}
&&  n (R_1 + R_2 - \epsilon_n )   \nonumber \\
   &\leq& I(W_1 ; \mathbf{Y}^n ) - I(W_1 ; \mathbf{T}^n \mathbf{Z}^n) + I(W_1 W_2;\mathbf{T}^n \mathbf{Z}^n) \\
&\leq& I(W_1 ; \mathbf{Y}^n Z^n_1 ) - I(W_1 ; \mathbf{T}^n \mathbf{Z}^n) + I(W_1 W_2; \mathbf{T}^n \mathbf{Z}^n)\qquad \\
&\overset{(a)}{=}& I(W_1 ;\mathbf{Y}^n |Z^n_1 ) - I(W_1 ; \mathbf{T}^n Z_2^n | Z_1^n)   \nonumber \\
&& \qquad  + I(W_1 W_2; \mathbf{T}^n Z_2^n | Z_1^n) + n\epsilon_n  \ , \ \ \ 
\end{IEEEeqnarray}
where $(a)$ follows from the secrecy constraint.  By standard manipulations, similarly to those used in the proof of the outer bound in Section~\ref{SumRatesProof}, write that:
\begin{IEEEeqnarray}{rCl}
&&  n (R_1 + R_2 - \epsilon_n )   \nonumber \\
  &\leq&  \sum^n_{i=1} \left[ I( W_1 \mathbf{T}^n_{i+1} ; Y_i |Z^n_1 \mathbf{Y}^{i-1} Z_{2, i+1}^n )  \right. \nonumber \\ 
  &&  \left. -  I( W_1 \mathbf{T}^n_{i+1} ; \mathbf{T}_i Z_{2,i}| Z^n_1 \mathbf{Y}^{i-1} Z_{2, i+1}^n ) \right. \nonumber \\ 
  &&    \vphantom{\sum^n } \left.   + I(W_1 W_2 \mathbf{T}^n_{i+1}Z_2^{i-1} ; \mathbf{T}_i |  Z_{2,i}  Z^n_1 \mathbf{Y}^{i-1} Z_{2, i+1}^n )\right]  \\
  &=& \sum^n_{i=1}\left[   I( W_1 \mathbf{T}^n_{i+1} ;  Y_{1,i}   Y_{2,i} |Z^n_1 \mathbf{Y}^{i-1} Z_{2, i+1}^n )  \right. \nonumber \\ 
  &&    \left.- I( W_1 \mathbf{T}^n_{i+1} ;  T_{1,i}  T_{2,i} Z_{2,i}| Z^n_1 \mathbf{Y}^{i-1} Z_{2, i+1}^n )  \right. \nonumber \\ 
  &&    \vphantom{\sum^n } \left.  + I(W_1 W_2 \mathbf{T}^n_{i+1}Z_2^{i-1} ;  T_{1,i}  T_{2,i} | Z_{2,i} Z^n_1 \mathbf{Y}^{i-1} Z_{2, i+1}^n )\right]   \\
  &=& \sum^n_{i=1} \left[ I( W_1 \mathbf{T}^n_{i+1} ;  Y_{2,i} |Z^n_1 \mathbf{Y}^{i-1} Z_{2, i+1}^n ) \right. \nonumber \\ 
  &&    \left. -  I( W_1 \mathbf{T}^n_{i+1} ;  T_{2,i} Z_{2,i}| Z^n_1 \mathbf{Y}^{i-1} Z_{2, i+1}^n )  \right. \nonumber \\ 
  &&    \vphantom{\sum^n } \left.  + I(W_1 W_2 \mathbf{T}^n_{i+1}Z_2^{i-1} ;   T_{2,i} | Z_{2,i} Z^n_1 \mathbf{Y}^{i-1} Z_{2, i+1}^n )  \right. \nonumber \\ 
  &&   \vphantom{\sum^n }  \left.  +  I( W_1 \mathbf{T}^n_{i+1} ;  Y_{1,i} | Y_{2,i} Z^n_1 \mathbf{Y}^{i-1} Z_{2, i+1}^n )  \right. \nonumber \\ 
  &&   \vphantom{\sum^n }  \left. - I( W_1 \mathbf{T}^n_{i+1} ;   T_{1,i} | T_{2,i} Z_{2,i}Z^n_1 \mathbf{Y}^{i-1} Z_{2, i+1}^n )  \right. \nonumber \\ 
  &&  \vphantom{\sum^n } \left.  + I(W_1 W_2 \mathbf{T}^n_{i+1}Z_2^{i-1} ;  T_{1,i} | T_{2,i}Z_{2,i} Z^n_1 \mathbf{Y}^{i-1} Z_{2, i+1}^n )\right]  \qquad    \\
  &\leq&  \sum^n_{i=1} \left[ I( W_1 \mathbf{T}^n_{i+1} ;  Y_{2,i} |Z^n_1 \mathbf{Y}^{i-1} Z_{2, i+1}^n )  \right. \nonumber \\ 
  &&  \left. -  I( W_1 \mathbf{T}^n_{i+1} ;  T_{2,i} Z_{2,i}| Z^n_1 \mathbf{Y}^{i-1} Z_{2, i+1}^n )  \right. \nonumber \\ 
  &&  \vphantom{\sum^n } \left. + I(X_{2,i} ;   T_{2,i} | Z_{2,i} Z^n_1 \mathbf{Y}^{i-1} Z_{2, i+1}^n )  \right. \nonumber \\ 
  &&  \vphantom{\sum^n } \left.  +  I( W_1 \mathbf{T}^n_{i+1} ; Y_{1,i} | Y_{2,i} Z^n_1 \mathbf{Y}^{i-1} Z_{2, i+1}^n )  \right. \nonumber \\ 
  && \vphantom{\sum^n } \left. - I( W_1 \mathbf{T}^n_{i+1} ;  T_{1,i} | T_{2,i} Z_{2,i}Z^n_1 \mathbf{Y}^{i-1} Z_{2, i+1}^n )   \right. \nonumber \\ 
  &&  \vphantom{\sum^n } \left.  + I(W_1 W_2 \mathbf{T}^n_{i+1}Z_2^{i-1} ;  {T_{1,i} } | T_{2,i}Z_{2,i} Z^n_1 \mathbf{Y}^{i-1} Z_{2, i+1}^n )  \right] \ . 
\end{IEEEeqnarray} 
On one hand, we observe that:
\begin{IEEEeqnarray}{rCl}
 && I( W_1 \mathbf{T}^n_{i+1} ;  Y_{1,i} | Y_{2,i} Z^n_1 \mathbf{Y}^{i-1} Z_{2, i+1}^n )    \nonumber \\ 
  && \quad   - I( W_1 \mathbf{T}^n_{i+1} ;  T_{1,i} | T_{2,i} Z_{2,i}Z^n_1 \mathbf{Y}^{i-1} Z_{2, i+1}^n )   \nonumber \\ 
  && \quad   + I(W_1 W_2 \mathbf{T}^n_{i+1}Z_2^{i-1} ;  T_{1,i} | T_{2,i}Z_{2,i} Z^n_1 \mathbf{Y}^{i-1} Z_{2, i+1}^n )   \nonumber\\ 
  &=& I( W_1 \mathbf{T}^n_{i+1}  ; Y_{1,i} | Y_{2,i} Z^n_1 \mathbf{Y}^{i-1} Z_{2, i+1}^n )   \nonumber \\ 
  && \quad   + I( W_2 Z_2^{i-1} ;  T_{1,i}   | T_{2,i}Z_{2,i} Z^n_1 \mathbf{Y}^{i-1} Z_{2, i+1}^n W_1 \mathbf{T}^n_{i+1} ) \,\,\, \,\,\, \,\,\, \,\,\,  \\
  &\overset{(a)}{=}& I( W_1 \mathbf{T}^n_{i+1}; Y_{1,i} | Y_{2,i} Z^n_1 \mathbf{Y}^{i-1} Z_{2, i+1}^n )    \nonumber \\ 
  && \quad   + I( W_2 Z_2^{i-1} ;  T_{1,i}   | T_{2,i} Z^n_1 \mathbf{Y}^{i-1} Z_{2, i+1}^n W_1 \mathbf{T}^n_{i+1} ) \\
  &\leq& I( W_1 \mathbf{T}^n_{i+1};   Y_{1,i}  | Y_{2,i} Z^n_1 \mathbf{Y}^{i-1} Z_{2, i+1}^n )    \nonumber \\ 
  && \quad   + I( X_{1,i} ;  T_{1,i} | T_{2,i} Z^n_1 \mathbf{Y}^{i-1} Z_{2, i+1}^n  W_1 \mathbf{T}^n_{i+1} ) \ , 
\end{IEEEeqnarray} 
where $(a)$ follows from that $Z_2$ is degraded respect to $T_2$. On the other hand, we have that:
\begin{IEEEeqnarray}{rcl}
&& I(X_{1,i}; T_{1,i} |  T_{2,i}  Z_{1,i}) \nonumber \\ 
& \overset{(a)}{=}& I(X_{1,i}; T_{1,i}  |Z_{1,i}) - I( T_{2,i} ; T_{1,i}|Z_{1,i} ) \\
& \overset{(b)}{\leq}& I(X_{1,i}; T_{1,i}|Z_{1,i} ) - I( Y_{2,i}  ; T_{1,i} |Z_{1,i}) \\
&=& I(X_{1,i}; T_{1,i} |  Y_{2,i} Z_{1,i}) \ , 
\end{IEEEeqnarray}
where $(a)$ and $(b)$ follow from the Markov chains: 
\begin{equation}
  (Y_{2,i}, T_{2,i}) \mkv X_{1,i} \mkv (Y_{1,i}, Z_{1,i})  
\end{equation}
and \begin{equation}
 (Y_{1,i},Z_{1,i})  \mkv X_{2,i}  \mkv(Y_{2,i}, T_{2,i}) \ , 
\end{equation}
and thus this implies that $T_2$ is less-noisy than $Y_2$. From this observation, we have:
\begin{IEEEeqnarray}{rCl}
 && I(W_1 \mathbf{T}^n_{i+1} ;   Y_{1,i}  | Y_{2,i} Z^n_1 \mathbf{Y}^{i-1} Z_{2, i+1}^n )    \nonumber \\ 
  && \qquad  + I( X_{1,i} ;  T_{1,i} | T_{2,i} Z^n_1 \mathbf{Y}^{i-1} Z_{2, i+1}^nW_1 \mathbf{T}^n_{i+1} ) \IEEEnonumber \\
  &\leq&  I( W_1 \mathbf{T}^n_{i+1} ;   Y_{1,i}  | Y_{2,i} Z^n_1 \mathbf{Y}^{i-1} Z_{2, i+1}^n )   \nonumber \\ 
  && \qquad  + I( X_{1,i} ;  T_{1,i} |  Y_{2,i} Z^n_1 \mathbf{Y}^{i-1} Z_{2, i+1}^n W_1 \mathbf{T}^n_{i+1}) \\
  &=& I( X_{1,i} ;  T_{1,i} | Y_{2,i} Z^n_1 \mathbf{Y}^{i-1} Z_{2, i+1}^n ) \\
  &\leq& I( X_{1,i} ;  T_{1,i} |  Z_{1,i} ) \ . 
\end{IEEEeqnarray} 
Then, letting $S_{2,i} = \mathbf{T}^n_{i+1}$, the resulting sum-rate reads as: 
\begin{IEEEeqnarray}{rCl}
 R_1 + R_2 &\leq& I(X_1; Y_1|Z_1) + I(U_2 S_2; Y_2|V_2)   \nonumber \\ 
  &&  - I(U_2 S_2; T_2 Z_2|V_2) + I(X_2;T_2 |Z_2 V_2)\ .\,\,\, 
\end{IEEEeqnarray}
The variable $S_2$ can be eliminated in a similar manner as we already did  in Section~\ref{EliminationS}.  Since, $Y_2$ is less-noisy than $Z_2$ and so is $T_1$ towards $Z_1$, then we can show the converse of the region by letting $U_2 \equiv (U_2,V_2)$ and $U_1 \equiv (U_1,V_1)$. 

\section*{Acknowledgment}
The authors are grateful to the Associate Editor Prof. Yingbin Liang and to anonymous reviewers for very constructive comments and suggestions on the earlier version of the paper, which has significantly improved its quality.

\bibliographystyle{IEEEtran}
\bibliography{FirstBibJabRef} 
%\bibliography{IEEEabrv,biblio}

%\begin{IEEEbiographynophoto}{Meryem Benammar}(S'11-M'14)
%received the B.S. and M.S.
%degrees in Electrical Engineering from CentraleSupelec in 2011, and the Ph.D. in Telecommunications from CentraleSup\'elec in 2014. Since 2014, she joined the Mathematical and Algorithmic Sciences Lab, Huawei Technologies, France, as a researcher. Her research interests  include network information theory, physical layer security and multi-terminal source coding. 
%\end{IEEEbiographynophoto}
%\begin{IEEEbiographynophoto}{Pablo Piantanida}
%(S'04-M'08) received the B.Sc. and M.Sc degrees (with honors) in Electrical 
%Engineering from the University of Buenos Aires (Argentina), in 2003, and the 
%Ph.D. from the Paris-Sud University (France) in 2007. In 2006, he has been with 
%the Department of Communications and Radio-Frequency Engineering at Vienna 
%University of Technology (Austria). Since October 2007 he is with the Department 
%of Telecommunications, CentraleSup{\'e}lec, as an Assistant Professor in network 
%information theory. His research interests include multi-terminal information 
%theory, Shannon theory, cooperative communications, physical-layer security and 
%distributed source coding.
%\end{IEEEbiographynophoto}

	\end{document}